\documentclass[journal]{IEEEtran}
\usepackage{cite}

\ifCLASSINFOpdf
\else
\fi

\usepackage{amsthm}
\usepackage{amsfonts}
\usepackage{mathtools}
\usepackage{algorithm}
\usepackage[noend]{algpseudocode}
\usepackage{float}
\usepackage{subfig}
\usepackage{multirow}
\usepackage{makecell}
\usepackage{xcolor}
\usepackage{booktabs}

\let\G\relax

\let\E\relax
\let\define\relax

\DeclareMathOperator{\R}{\mathcal{R}}
\DeclareMathOperator{\G}{\mathcal{G}}
\DeclareMathOperator{\E}{\mathbb{E}}

\DeclareMathOperator{\T}{\mathcal{T}}

\DeclareMathOperator{\M}{\mathcal{M}}

\DeclareMathOperator{\lo}{\text{lower}}

\DeclareMathOperator{\TIME}{TIME}

\DeclareMathOperator*{\argmax}{arg\,max}
\DeclareMathOperator*{\argmin}{arg\,min}
\DeclareMathOperator{\define}{\coloneqq}
\DeclareMathOperator{\poly}{poly}

\newtheorem{theorem}{\textbf{Theorem}}
\newtheorem{lemma}{\textbf{Lemma}}
\newtheorem{corollary}{\textbf{Corollary}}

\theoremstyle{definition}

\newtheorem{definition}{\textbf{Definition}}
\newtheorem{problem}{\textbf{Problem}}
\newtheorem{remark}{\textbf{Remark}}
\newtheorem{claim}{\textbf{Claim}}

\newcommand{\RNum}[1]{\uppercase\expandafter{\romannumeral #1\relax}}

\begin{document}
%
\title{On Multi-Cascade Influence Maximization: Model, Hardness and Algorithmic Framework}
%
%
%

\author{Guangmo Tong,~\IEEEmembership{Member,~IEEE,}
        Ruiqi Wang,
        and Zheng Dong,~\IEEEmembership{Member,~IEEE}
\thanks{G. Tong and R. Wang are with the Department
of Computer and Information Science, University of Delaware, Newark,
DE, 19711, USA,  e-mail: \{amotong, wangrq\}@udel.edu.}
\thanks{Z. Dong is with Wayne State University, Detroit, MI ,48202, USA. e-mail: dong@wayne.edu.}}

%
%

\markboth{Journal of \LaTeX\ Class Files,~Vol.~14, No.~8, August~2015}%
{Shell \MakeLowercase{\textit{et al.}}: Bare Demo of IEEEtran.cls for IEEE Journals}
%



\maketitle

\begin{abstract}
This paper studies the multi-cascade influence maximization problem, which explores strategies for launching one information cascade in a social network with multiple existing cascades. With natural extensions to the classic models, we first propose the independent multi-cascade model where the diffusion process is governed by the so-called activation function. We show that the proposed model is sufficiently flexible as it generalizes most of the existing cascade-based models. We then study the multi-cascade influence maximization problem under the designed model and provide approximation hardness under common complexity assumptions, namely Exponential Time Hypothesis and $NP \subseteq DTIME(n^{\poly \log n})$. Given the hardness results, we build a framework for designing heuristic seed selection algorithms with a testable data-dependent approximation ratio. The designed algorithm leverages upper and lower bounds, which reveal the key combinatorial structure behind the multi-cascade influence maximization problem. The performance of the framework is theoretically analyzed and practically evaluated through extensive simulations. The superiority of the proposed solution is supported by encouraging experimental results, in terms of effectiveness and efficiency. 
\end{abstract}


%
\IEEEpeerreviewmaketitle

\section{Introduction}
\IEEEPARstart{I}{nfluence} Maximization (\textbf{IM}) considers the problem of selecting seed nodes for an information cascade such that the total influence can be maximized. IM finds applications in various domains, such as viral marketing and epidemic control, and it thus has drawn tremendous attention \cite{li2018influence,zhang2014recent,sun2011survey,aslay2018influence}. In the seminal work of Kempe, Tardos and Kleinberg \cite{kempe2003maximizing}, two basic operational models, Independent Cascade (\textbf{IC}) model and Linear Threshold (\textbf{LT}) model, were proposed for modeling information diffusion, where the former focuses on the peer-to-peer communication while the latter concerns the accumulative influence from neighbors. Motivated by IM, more realistic scenarios were studied, among which an important one is the multi-cascade influence maximization where we assume that there are multiple existing cascades. Indeed, one can find examples that multiple companies compete for brand awareness on Facebook or that misinformation and its debunking information fight against each other on Twitter. In this paper, we study the \textit{Multi-cascade Influence Maximization} (\textbf{MIM}) problem. Our work consists of four parts: model design, hardness analysis, algorithm design, and experimental study.

\textbf{Model Design.} The start point is to build a multi-cascade diffusion model where the key setting is to define the correlation between different cascades during the diffusion process. Under the classic IC model for single-cascade diffusion, each activated user has one chance to activate each of their inactive neighbors, which defines the manner in which the information propagates between a pair of users. When there are multiple cascades, we need to further define the outcome when two or more cascades reach one user at the same time. Following such a prototype, we aim at designing multi-cascade diffusion models in this paper. On the one hand, there exist a few models for multi-cascade diffusion, but they are either competition-oblivious \cite{datta2010viral} by assuming that the diffusion of different cascades is independent, or they are designed based on particular settings such as the order-based models (e.g., \cite{li2014polarity,chen2011influence,tong2018distributed}) and priority-based models (e.g., \cite{budak2011limiting,tong2018misinformation}). These models are sufficiently simple so as to admit constant approximation algorithms for MIM, but they fail to capture the complex correlations between cascades. On the other hand, sophisticated models such as deep neural networks are extremely expressive, while they completely cloak the combinatorial structure. In this paper, we propose the concept of \textit{activation function}, which in general specifies how a user is activated if they receive multiple cascades simultaneously. Based on the activation function, we build the Independent Multi-Cascade (\textbf{IMC}) model. The activation function is not required to follow a certain form, and it can be flexibly realized for a specific application or learned from real data. As shown later, the proposed IMC model generalizes a few existing multi-cascade diffusion models.

\textbf{Hardness Analysis.} The IM problem is well understood credited to the recent advances starting from Kempe \textit{et al.} \cite{kempe2003maximizing}. More specifically, IM is shown to be NP-hard, and it admits constant approximation algorithms due to its submodularity. Even though its objective function is \#P-hard to compute \cite{chen2010scalable}, near-optimal algorithms have been designed and then improved significantly \cite{borgs2014maximizing,tang2014influence,tang2015influence,nguyen2016stop}. MIM has also been studied in the existing works, but very few hardness results are known in addition to being NP-hard to solve optimally. In this paper, we focus on the approximation hardness of MIM. Based on the IMC model, we consider the MIM problem under two formulations \textbf{Max-MIM} and \textbf{Min-MIM}, where Max-MIM maximizes the influence of the new cascade while Min-MIM minimizes the users that are not aware of the new cascade. For the Max-MIM problem, we prove that no  $n^{1/\poly \log \log n}$-approximation exists unless the Exponential Time Hypothesis (\textbf{ETH}) is false. For the Min-MIM problem, we prove that it cannot be approximated within a factor of $\Omega(2^{\log^{1-\epsilon}n^3-\log n})$ unless $NP \subseteq DTIME(n^{\poly \log n})$. 

\textbf{Algorithm Design.} The IM problem is currently solved by two main techniques: submodular maximization and reverse sampling. The MIM problem can be submodular under simple settings (e.g., two-cascade diffusion case), but it is not true for the general case, as observed in earlier works (e.g., \cite{lu2015competition,tong2018misinformation}). Therefore, existing techniques do not trivially apply to the multi-cascade case, and our objective herein is to design scalable algorithms for solving MIM under the general IMC model. Given the hardness results, we focus on heuristic algorithms and propose the Reverse Sandwich (\textbf{RS}) algorithm designed through a novel reverse sandwich sampling method. The key findings in our framework are the upper and lower bounds of the objective function which are simple but applicable to any activation function. More importantly, the derived bounds are submodular and therefore can be effectively approximated. Given certain parameters $\epsilon \in (0,1)$ and $l>1$, with probability at least $1-n^{-l}$ the proposed algorithm runs in $O(\frac{(m+n)(l+k)\ln n}{\epsilon^2})$ and provides an approximation ratio of $(1-\epsilon)\cdot \gamma(\Psi_l)-\epsilon$ where $\gamma(\Psi_l)$ is data-dependent and upper bounded by $1-1/e$.\footnote{It is called as data-dependent if it relies on the structure of the graph rather than solely on the number of nodes and edges. The formal definition of $\gamma(\Psi_l)$ is provided in later Sec. \ref{alg: framework}.} The ratio $\gamma(\Psi_l)$  can be evaluated for each produced solution and therefore provides real-time testable performance bound.

\textbf{Experiments.} We experimentally examine the performance of the proposed algorithm with different choices of activation functions on real-world datasets. The results have demonstrated the superiority of the RS algorithm in terms of both efficiency and effectiveness. In particular, the data-dependent ratio $\gamma(\Psi_l)$ is consistently close to $1-1/e$, implying that the approximation ratio is practically near-constant. Furthermore, the RS algorithm outperforms the simple greedy algorithm and other baselines by a significant margin, and it is scalable to handle large datasets. 

\textbf{Roadmap.} We survey the related work in Sec. \ref{sec: relate}. The system model and problem formulation are given in Sec. \ref{sec: pre}. In Sec. \ref{sec: hardness}, we analyze the hardness of MIM. The algorithmic framework is given in Sec. \ref{sec: framework}. Further discussions are put into Sec. \ref{sec: discuss}. In Sec. \ref{sec: exp}, we present the experimental results. Sec. \ref{sec: con} concludes this paper. The proofs can be found in the appendix.

%

\section{Related Work}
\label{sec: relate}

\textbf{Influence Maximization and Reverse Sampling.} IM was proposed by Kempe \textit{et al.} \cite{kempe2003maximizing} where the IC and LT model were formulated for information diffusion. It is shown in \cite{kempe2003maximizing} that IM as a combinatorial optimization problem is monotone submodular and therefore the simple greedy algorithm provides a $(1-1/e)$-approximation \cite{feige1998threshold}. However, the objective function of IM under the classic models is \#P-hard to compute \cite{chen2010scalable} and efficient heuristics were designed later by various researchers (e.g., \cite{li2018influence, zhang2014recent}). One breakthrough to overcome the \#P-hardness was made by Borgs \textit{et al.} \cite{borgs2014maximizing} where the reverse sampling method was designed to efficiently estimate the objective function. Based on the reverse sampling method, advanced stochastic algorithms for solving IM were designed (\cite{tang2014influence,tang2015influence,nguyen2016stop}) which have significantly reduced the running time without sacrificing the performance bound. Even though these techniques do not directly apply to MIM, our algorithm for MIM in this paper still benefits from the reverse sampling method as well as its improvements.

\textbf{Multi-cascade Diffusion.} As an immediate generalization of IM, multi-cascade influence maximization has also been extensively studied. Earlier literature has worked on the case when there are two information cascades (e.g., \cite{budak2011limiting,he2012influence,fan2013least,lin2015analyzing,tong2017efficient,li2014polarity, li2018dominated}) in which the objective function remains monotone and submodular and therefore constant approximations are obtainable. Later in \cite{lu2015competition}, the authors considered the comparative influence diffusion where one node can adopt multiple cascades, and showed that the MIM problem is not submodular. When multiple cascades exist, the influence maximization problem can be elusive as even not being monotone \cite{tong2018distributed}, showing that carelessly selecting seed nodes can be harmful. However, to the best of our knowledge, there is no hardness of approximation known for MIM under any of the existing multi-cascade diffusion models, which is the gap we attempt to fill in this paper. For multi-cascade diffusion, there exists literature concerning game-theoretical analysis  (e.g., \cite{molinero2015cooperation, clark2011maximizing, fazeli2012game, tzoumas2012game, apt2011diffusion, li2015getreal}), seed minimization problems (e.g., \cite{shirazipourazad2012influence, zhu2016minimum, bozorgi2017community,lu2013bang}), threshold-based models  (e.g., \cite{apt2014social,molinero2015cooperation,pathak2010generalized,he2012influence, borodin2010threshold,bozorgi2017community,lin2015learning, litou2017influence}) or temporal models \cite{zarezade2017correlated}, which are however not closely related to our work. In this paper, we focus on cascade-based diffusion models where  a user can be activated for at most once.

\textbf{Sandwich Approximation Strategy.} We in this paper leverage the sandwich approximation strategy proposed by Lu \textit{et al.} \cite{lu2015competition}, which utilizes the upper and lower bounds of the objective function. Such a method have been adopted in a few following works (e.g., \cite{garimella2017balancing, wang2017activity, tong2018misinformation}) where the bounds are model specific and furthermore cannot produce scalable algorithms. Different from the existing works, our bounds in this paper are observed during the reverse sampling phase, and they apply to the IMC model with arbitrary activation functions.

%
%
%
%
%
%
%
\section{Preliminaries}
\label{sec: pre}

\subsection{Diffusion Model}
In this section, we present the Independent Multi-Cascade (IMC) model which is an extension to the classic IC model \cite{kempe2003maximizing}. The social network is given by a directed graph $G=(V,E)$ where associated with each ordered pair of nodes $u$ and $v$ there is a propagation probability $p_{(u,v)} \in (0,1]$. Let $n=|V|$ and $m=|E|$ be number of nodes and edges, respectively. For a graph $G$, we use the notation $V(G)$ (resp., $E(G)$) to denote its node set (resp., edge set). We use $C$ to denote the set of the cascades and speak of a node $v$ as being \textit{$c$-active} for some cascade $c \in C$ if $v$ is activated by cascade $c$. 
For a certain time step $t$, we use  $\pi_t(v)$ to denote the state of node $v$. If $v$ is not activated by any cascade, we define $\pi_t(v)=\emptyset$. Each cascade $c \in C$ starts to spread from its seed set $\tau(c) \subseteq V$, and the nodes in $\tau(c)$ are the first being $c$-active. When a node $v$ is $c$-active for some cascade $c \in C$, they have one chance to activate each inactivate out-neighbor $v$ with a success probability of $p_{(u,v)}$ and $v$ will also be $c$-active if the activation succeeds. We assume that one node can be activated for at most once, and therefore they will be typically activated by the cascade arriving first. However, when there are multiple cascades, one node can be activated by neighbors belonging to different cascades at the same time. The core part of a multi-cascade model is to define the behavior of the diffusion process in such a scenario. For such purposes, we propose the following concept.

\begin{definition}[\textbf{Activation Function}]
\label{def: pri_function}
Let $V \times C$ be the space consisting of pairs of node and cascade. For each node $v$, there is an \textit{activation function} $F_v: 2^{V \times C} \rightarrow C$. The activation function controls the diffusion process as follows:
\begin{itemize}
\item (a) When a node $v$ is selected as a seed node by a set of cascades $C^{'} \subseteq C$, $v$ will be $F_v(O)$-active where $O=\{(v,c)|c\in C^{'}\} \subseteq V \times C$ and $F_v(O) \in C^{'}$, which defines the initial states.
\item  (b) When $v$ is for the first time activated by a set of neighbors $A$ at a certain time step $t$, $v$ will be $F_v(O)$-active where $O=\{(u,\pi_{t-1}(u))|u\in A\} \subseteq V \times C$ and $F_v(O) \in \{\pi_{t-1}(u)| u\in A\}$, which defines the cascade behavior during the diffusion.
\end{itemize}
 
\end{definition}

Incorporating activation function into the IC model, we are ready to present the diffusion process under the IMC model.

\begin{definition}[\textbf{Diffusion Process}] Given the seed sets of the cascades, the diffusion process goes as follows:
\begin{itemize}
\item Step $0$. If a node $v$ is selected as a seed node by a set $C^{'}$ of cascades, its state is determined according to Def. \ref{def: pri_function} (a). Other nodes will remain to be $\emptyset$-active.
\item Step $t$. Each node $u$ activated in time $t-1$ to activate each of $u$'s inactivate out-neighbors, say $v$, with a success probability of $p_{(u,v)}$. If a node $v$ is successfully activated by one or more neighbors, their state is determined according to Def. \ref{def: pri_function} (b).
\item The process terminates when no node can be further activated.
\end{itemize}

\end{definition}

\begin{figure}[t]
\begin{center}
\includegraphics[width=0.40\textwidth]{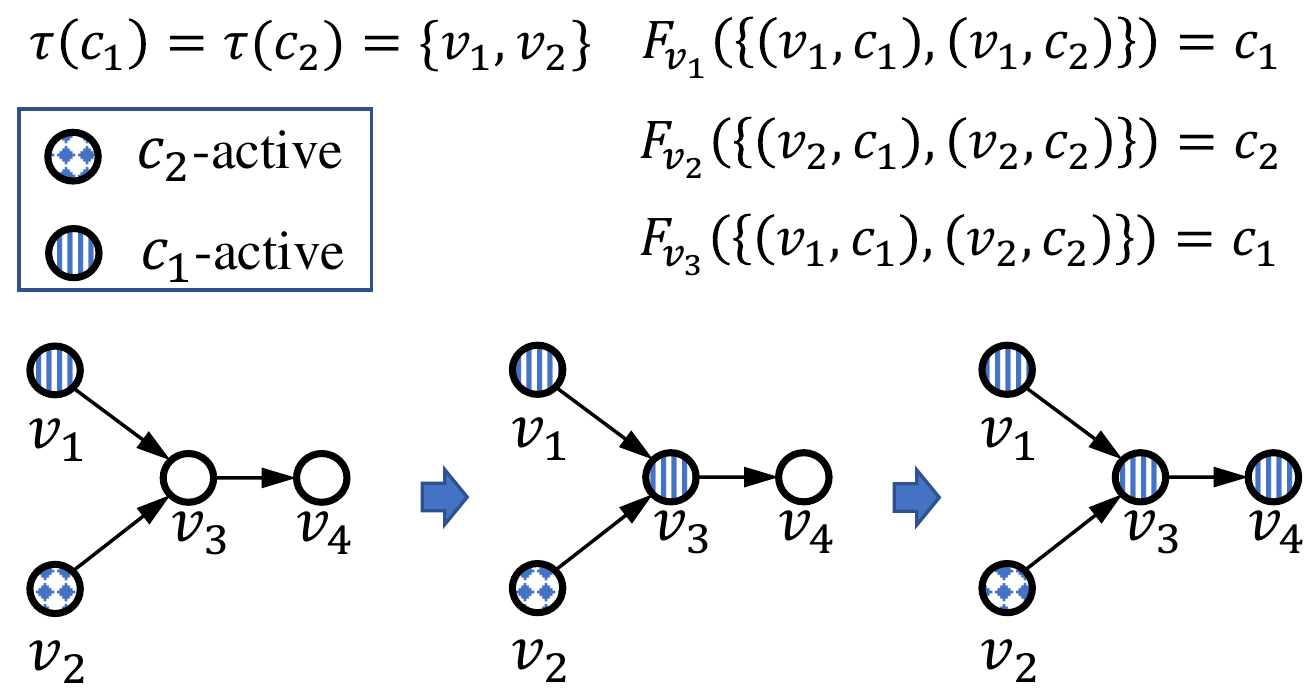} 
\end{center} 
\caption{Diffusion Process Example.  In this example, we have two cascades with the same seed set $\{v_1,v_2\}$ and we assume the propagation probability is equal to 1 for each edge. According to the activation function given in the figure, $v_1$ and $v_2$ will be $c_1$-active and $c_2$-active, respectively, and, $v_3$ and $v_4$ will then become $c_1$-active.}
\label{fig: example}
\vspace{-3mm}
\end{figure}

A toy example for illustration is given in Fig. \ref{fig: example}

\begin{remark}
It is possible to consider probabilistic activation functions, where $F_v$ specifies a probability distribution over the coming cascades. In this paper, we will focus on deterministic activation functions for the convenience of analysis, and our analysis applies to the probabilistic case as well. In another issue, the diffusion process generally exhibits competition as we assume a node never changes their state, but the activation function can also be implemented to describe the cooperation between cascades to a certain degree. For example, for a node $v$ with in-neighbors $v_1, v_2$ and $v_3$, we may have the activation function for three cascades $c_1, c_2$ and $c_3$, such that $F_v(\{(v_1,c_1), (v_2,c_2)\})=c_1$ and  $F_v(\{(v_1,c_1), (v_2,c_2), (v_3,c_3)\})=c_2$, showing that the exposure of $c_3$ to $v$ can help cascade $c_2$ win against $c_1$. In the context of economics, $c_3$ might be a complementary good to $c_2$. Note that the activation function is not required to be defined over the entire $2^{V \times C}$ as not all the subsets are valid inputs. Finally, an important setting is that $F_v(O) \in \{\pi_{t-1}(u)| u\in A\}$, implying that $v$ must be activated by the cascade that reaches them. In other models (e.g., \cite{li2014polarity,li2015getreal}), it is allowed that  $F_v(O) \notin \{\pi_{t-1}(u)| u\in A\}$ to which the analysis in this paper does not apply.
\end{remark}

\textbf{Special Models.} The IMC model reduces to the classic IC model when there is only one cascade. For the multi-cascade case, the activation function in its general form does not assume any specific rule and thus it generalizes many of the existing models, such as the order-based models \cite{li2014polarity,chen2011influence,tong2018distributed} and cascade-priority-based models \cite{budak2011limiting,tong2018misinformation} of which the settings can be equivalently described by the activation function. We herein mention two types of activation functions that are commonly seen in real social networks.

\begin{definition}[Cascade-based Activation Function]
\label{def: cascade_function}
A special class of activation functions can be obtained by solely considering the characteristics of the cascades and ignoring the information on the neighbor side. 
For such a case, the activation function can be simplified as $F_v: 2^C \rightarrow C$ taking only the coming cascades as the input. We may further simplify it by assuming a total order of the cascades at each node, and define $F_v: C \rightarrow \{1,...,|C|\}$ with $F_v(c_1)\neq F_v(c_2)$ for two cascades $c_1$ and $c_2$, meaning that $v$ will be activated by the cascade with the largest $F_v(c)$. Such a setting is termed as \textit{total-order-based activation function}.
\end{definition}

\begin{definition}[Neighbor-based Activation Function]
\label{def: neighbor_function}
Another type of activation function is designed with the consideration of the degree of friendship. It is intuitive that a node will be convinced by the best friend if they receive different cascades from multiples friends, and therefore the activation functions can be implemented as a total order among the in-neighbors. Such activation functions generalize the models in \cite{li2014polarity,tong2018distributed}. In addition, it is intuitive to assume that a user will be activated by the cascade that has been adopted by the majority or weighted majority of the neighbors.
\end{definition}

\subsection{Problem Formulation}
Given the seed sets of the existing cascades and an integer $k \in \mathbb{Z}^+$, our goal is to launch a new cascade with $k$ seed nodes selected from a candidate set $V_c \subseteq V$ such that its total influence is maximized. Without loss of generality, we assume $V_c = V$ unless otherwise stated. We use $C_e=\{c_1,...,c_L\}$ to denote the existing $L\in \mathbb{Z}^+$ cascades with known seed sets $\{\tau(c_i)\}_{i=1}^{L}$ and use $c_{new}$ to denote the newly introduced cascade. Therefore, the total cascade set is $C=\{c_1,...,c_L, c_{new}\}$. We assume that the activation function $F_v$ defined over $C$ for each $v$ is known to us, and, we use $f(S)$ to denote the expected influence of $c_{new}$ when $S \subseteq V$ is selected as its seed set. The considered problem is formally described as follows.

\begin{problem}[\textbf{Multi-cascade Influence Maximization (MIM) Problem}]
\label{problem: main}
Compute 
\[S_{opt} \define \argmax_{\{S \subseteq V_c, |S|= k\}}f(S),\]
or equivalently,
\[S_{opt} \define \argmin_{\{S \subseteq V_c, |S|= k\}} \overline{f}(S),\]
where $\overline{f}(S)$ is the expected number of the nodes that are not $c_{new}$-active. We use the notation $\M$ to denote an instance of MIM, and denote the problems with the above two objectives as Max-MIM and Min-MIM, respectively. 
\end{problem}

\begin{remark}
Some applications such as viral marketing are more suitable to be considered as Max-MIM because the profit is counted in terms of the influenced nodes, while in other scenarios, such as propagating critical information, Min-MIM is more appropriate because the loss therein is usually measured by the number of the users not receiving the information. These two problems are equivalent in terms of the optimal solution while they may have different inapproximability. Due to the linearly of expectation, it useful to decompose $f(S)$ as $f(S)=\sum_{v \in V} f_v(S)$ where $f_v(S)$ is the probability that node $v$ is $c_{new}$-active. 
\end{remark}


\section{Hardness Analysis}  
\label{sec: hardness}
This section establishes the hardness of Min-MIM and Max-MIM. We omit to show their NP-hardness as stronger hardness results will be provided.

\subsection{Hardness of Max-MIM}
The hardness of Max-MIM is derived from the Densest k-Subgraph (\textbf{DkS}) problem.
\begin{problem}[Densest k-subgraph]
Given an undirected graph $\overline{G}=(\overline{V},\overline{E})$ and $\overline{k}\in \mathbb{Z}^+$, find the subgraph of $\overline{k}$ nodes with the maximum number of edges. 
\end{problem}

\begin{lemma} 
\label{lemma: max_mim}
DkS can be approximated within a factor of $2 \cdot \alpha$ if Max-MIM can be approximated within a factor of $\alpha$ for each $\alpha>1$.
\end{lemma}

\begin{proof}
See Appendix \ref{proof: lemma: max_mim}.
\end{proof}

There are several existing hardness results \cite{feige2002relations, alon2011inapproximability, bhaskara2012polynomial} of DkS and we herein mention one recent ETH-hardness given by Manurangsi \cite{manurangsi2017almost}.
\begin{lemma}[\cite{manurangsi2017almost}]
\label{lemma: dks}
Unless the Exponential Time Hypothesis (ETH) is false, there exists no polynomial-time approximation algorithm for the DkS problem with a ratio of $n^{1/\poly \log \log n}$.
\end{lemma}

According to Lemma \ref{lemma: max_mim}, the approximation hardnesses of DkS immediately apply to Max-MIM, motivating us to design data-dependent approximation strategies.


\subsection{Hardness of Min-MIM}
The hardness result for the Min-MIM problem is established on a reduction from the k-positive-negative partial set cover (k$\pm$PSC) problem.

\begin{problem}[k$\pm$PSC problem]
An instance of k$\pm$PSC is a triplet $(X,Y,\Phi)$ with an integer $\overline{k} \in \mathbb{Z}^+$, where $X$ and $Y$ are two sets of elements with $X \cap Y=\emptyset$, and $\Phi=\{{\phi}_1,...,{\phi}_q\} \subseteq 2^{X \cup Y}$ is a collection of $q \in \mathbb{Z}^+$subsets over $X \cup Y$. For each ${\Phi}^* \subseteq \Phi$, its cost is defined as \[cost({\Phi}^*)=|X\setminus (\cup_{\phi \in {\Phi}^*} \phi)|+|Y\cap (\cup_{\phi \in {\Phi}^*} \phi)|.\] The k$\pm$PSC problem seeks for a ${\Phi}^* \subseteq \Phi$ with $|{\Phi}^*|=\overline{k}$ such that the cost is minimized.
\end{problem}

The following hardness of k$\pm$PSC follows fairly directly from Miettinen \cite{miettinen2008positive}. 
\begin{lemma}[\cite{miettinen2008positive}]
\label{lemma: kpsc}
Unless $NP \subseteq DTIME(n^{\poly \log n})$, there exists no polynomial-time approximation algorithm for the $k\pm$PSC problem with a ratio of $\Omega(2^{\log^{1-\epsilon}q})$ for each $\epsilon>0$.
\end{lemma}


Now we relate Min-MIN and k$\pm$PSC.
\begin{lemma}
\label{lemma: reduction}
k$\pm$PSC is approximable to within a factor of $\Omega(\alpha(q) \cdot q)$ if Min-MIM is approximable to within a factor of $\alpha(|V_c|)$ for each function $\alpha(\cdot)>1$.
\end{lemma}
\begin{proof}
See Appendix \ref{proof: lemma: reduction}.
\end{proof}

With the above two lemmas, we have the following hardness result showing the strong inapproximability of Min-MIM.
\begin{theorem}
\label{theorem: hardness}
Min-MIM cannot be approximated within a factor of $\Omega(2^{\log^{1-\epsilon}n^3-\log n})$ unless $NP$ belongs to $DTIME(n^{\poly \log n})$.\footnote{$DTIME(n^{\poly \log n})$ consists of the decision problems that can be solved in $O(n^{\poly \log n})$. Note that here $n$ is not the number of the nodes of the social network graph.}
\end{theorem}
\begin{proof}
According to Lemmas \ref{lemma: kpsc} and $\ref{lemma: reduction}$, one has  $\alpha(q) \cdot q=\Omega(2^{\log^{1-\epsilon}q^4})$ and consequently we have that $\alpha(q)=\Omega(2^{\log^{1-\epsilon}q^3-\log q})$. Therefore, $\alpha(|V_c|)$ is asymptotically bounded by $2^{\log^{1-\epsilon}n^3-\log n}$ since $|V_c|=q\leq n$.
\end{proof}

The results in this section indicate constant approximation algorithms are not possible for MIM under common complexity assumptions, but they do not rule out efficient heuristics for the average case. In the next section, we will present an algorithmic framework that can solve MIM with a data-dependent approximation ratio.
\section{Algorithmic Framework}
\label{sec: framework}
A well-known difficulty regarding the cascade model is that the objective function is \#P-hard to compute, and it is typically addressed by scholastic optimization with proper sampling method. On the other hand, as shown in Sec. \ref{sec: hardness}, algorithms with worst-case approximation guarantees is pessimistic, and therefore we will pursue a data-dependent appropriation ratio. In this section, we provide an algorithm designed through a novel reverse sandwich sampling method with a data-dependent ratio which can be easily tested.  


\subsection{Reverse Sandwich Sampling Method}
\label{subsec: sampling}
The goal of sampling is to generate $l$ samples which can be used to design an estimator $\tilde{f}_l$ of $f$ such that $\lim_{l\rightarrow \infty}|\tilde{f}_l(S)-f(S)|=0$ for each $S \subseteq V_c$, and $\tilde{f}_l(S)$ can be easily computed and optimized. For such a purpose, we start by considering how to estimate the probability that an individual node $v$ will be $c_{new}$-active. Given a certain node $v$, let us consider the following sampling process.

\begin{algorithm}[t]
\caption{{RR-tuple of $v$} $(\M, v)$}\label{alg: rrtuple_v}
\begin{algorithmic}[1]
\State \textbf{Input:} an instance $\M$ of MIM and a node $v$; 
\State \textbf{Output:} $R_v=(G_v, \overline{S}_v, \underline{S}_v)$; 
\State Initialize $G_v$ with $V(G_v)=\{v\}$ and $E(G_v)=\emptyset$;  
\State $V_t \leftarrow \{v\}$, $\tau^* \leftarrow \cup_{c \in C_e} \tau(c)$, $\overline{S}_v \leftarrow \emptyset$, $\underline{S}_v \leftarrow \emptyset$; 
\While {true}
	\If {$V_t=\emptyset$}
	\State \textbf{Return} $R_v=(G_v, \overline{S}_v, \underline{S}_v)$;
	\EndIf 
	\If {$V_t \cap \tau^* \neq \emptyset$}
			\State $\overline{S}_v \leftarrow \overline{S}_v \cup V_{t}$;
		\State \textbf{Return} $R_v=(G_v, \overline{S}_v, \underline{S}_v)$;
	\EndIf 
	\State $\overline{S}_v \leftarrow \overline{S}_v \cup V_{t}$, $\underline{S}_v \leftarrow \underline{S}_v \cup V_{t}$; 
	\State $V_r \leftarrow V \setminus V(G_v)$; 
	\State $E^* \leftarrow \{(u_1, u_2)\in E| u_1 \in V_r, u_2 \in V_t \}$; $V_t \leftarrow \emptyset$;
	\For {each edge $(u_1,u_2) \in E^*$} 
	\State  $rand \leftarrow \mathcal{U}(0,1)$; \Comment{(0,1) uniform distribution}; 
	\If {$rand  \leq p_{(u_1,u_2)}$}
	\State $V_{t} \leftarrow V_{t} \cup \{u_1\}$;
	\State $V(G_v) \leftarrow V(G_v) \cup \{u_1\}$; 
	\State $E(G_v) \leftarrow E(G_v) \cup \{(u_1,u_2)\}$;
	\EndIf 
	\EndFor
\EndWhile
\end{algorithmic}
\end{algorithm}

\begin{definition}[RR-tuple of $v$]
\label{def: rr_tuple_v}
As shown in Alg. \ref{alg: rrtuple_v}, given a node $v$, we simulate the diffusion process in a reverse direction from $v$ in the manner of BFS (Line 5 to 19) until one seed of the existing cascades is reached (Line 8) or no node can be further activated (Line 6). With each run of Alg. \ref{alg: rrtuple_v}, we obtain a subgraph $G_v$ of which the nodes and edges are collected in Line 18 and 19, respectively. In addition, we obtain two subsets $\overline{S}_v$ and $\underline{S}_v$ of nodes where $\overline{S}_v$ is identical to $V(G_v)$ and $\underline{S}_v$ is the $\overline{S}_v$ with the removal of the nodes, if any, encountered during the last iteration of the WHILE loop. We use the notation $R_v=(G_v, \underline{S}_v, \overline{S}_v)$ to denote one sample returned by Alg. \ref{alg: rrtuple_v}, termed as an \textit{RR-tuple of $v$}. 
\end{definition}

For each $R_v=(G_v, \underline{S}_v, \overline{S}_v)$, let us consider an MIM instance 
$\M_{R_v}$ with graph $G_v$ where (a) $V(G_v) \cap \tau(c)$ is the seed set of each of the existing cascade $c \in C_e$, (b) $p_{e}=1$ for each edge $e$ in $G_v$, and (c) the activation functions remain the same as in $\M$. For a seed set $S$ of cascade $c_{new}$, we use the indicator $g(S ,R_v)$ to denote if $v$ will be $c_{new}$-active in $\M_{R_v}$, i.e., 
\[g(S,R_v) \define
  \begin{cases}
  1 &  \hspace{0mm} \hspace{-0.5mm} \text{$v$ is $c_{new}$-active under $S$ in $\M_{R_v}$} \\
  0 & \hspace{0mm} \hspace{-0.5mm} \text{otherwise } 
  \end{cases}.\]
Furthermore, we consider two variables

\[\overline{g}(S ,R_v) \define
  \begin{cases}
  1 &  \hspace{0mm} \hspace{-0.5mm} \text{$S \cap \overline{S}_v \neq \emptyset$} \\
  0 & \hspace{0mm} \hspace{-0.5mm} \text{otherwise } 
  \end{cases},\]
and 
\[\underline{g}(S,R_v) \define
\begin{cases}
1 &  \hspace{0mm} \hspace{-0.5mm} \text{$S \cap \underline{S}_v \neq \emptyset$} \\
0 & \hspace{0mm} \hspace{-0.5mm} \text{otherwise } 
\end{cases}.\] 
For a fixed $S$, $g(S,R_v)$, $\overline{g}(S,R_v)$ and $\underline{g}(S,R_v)$ are 
random variables as the sampling process is stochastic. We use $\R_v$ to denote a random RR-tuple of $v$. The next theorem is a key result, which gives the upper and lower bounds on the probability that an individual node can be $c_{new}$-active.

\begin{theorem}
\label{theorem: key_1}
For each $S \subseteq V$, we have
\begin{equation}
\E[\underline{g}(S , \R_v)] \leq \E[g(S ,\R_v)]=f_v(S)\leq E[\overline{g}(S ,\R_v)]
\end{equation}
\end{theorem}
\begin{proof}[Proof (Sketch)]
The rationale behind this theorem is as follows. Each $R_v=(G_v, \underline{S}_v, \overline{S}_v)$ returned by Alg. \ref{alg: rrtuple_v} corresponds to some realization\footnote{The formal definition of realization is given in the complete proof.} of the real diffusion process. For the part $\E[g(S ,\R_v)]=f_v(S)$, the proof follows from the idea of reverse sampling by showing that the nodes searched in Alg. \ref{alg: rrtuple_v} contains all the candidate seed nodes that can make $v$ become $c_{new}$-active in the realization corresponding to $R_v$. Furthermore, $S \cap \overline{S}_v \neq \emptyset$ is a necessary condition for making  $v$ be $c_{new}$-active in $\M_{R_v}$ under $S$, while $S \cap \underline{S}_v \neq \emptyset$ is a sufficient one. An illustration is given in Fig. \ref{fig: reverse}.
See Appendix \ref{proof: theorem: key_1} for the complete proof.
\end{proof}

\begin{algorithm}[t]
\caption{RR-tuple $(\M)$}
\label{alg: rrtuple}
\begin{algorithmic}[1]
\State \textbf{Input:} an instance $\M$ of MIM; 
\State \textbf{Output:} an RR-tuple.
\State Select a node $v$ from $V$ uniformly at random.
\State Generate an RR-tuple $R_v$ of $v$ by Alg. \ref{alg: rrtuple_v}.
\State return $R_v$; 
\end{algorithmic}
\end{algorithm}

\begin{figure}[t]
\begin{center}
\includegraphics[width=0.45\textwidth]{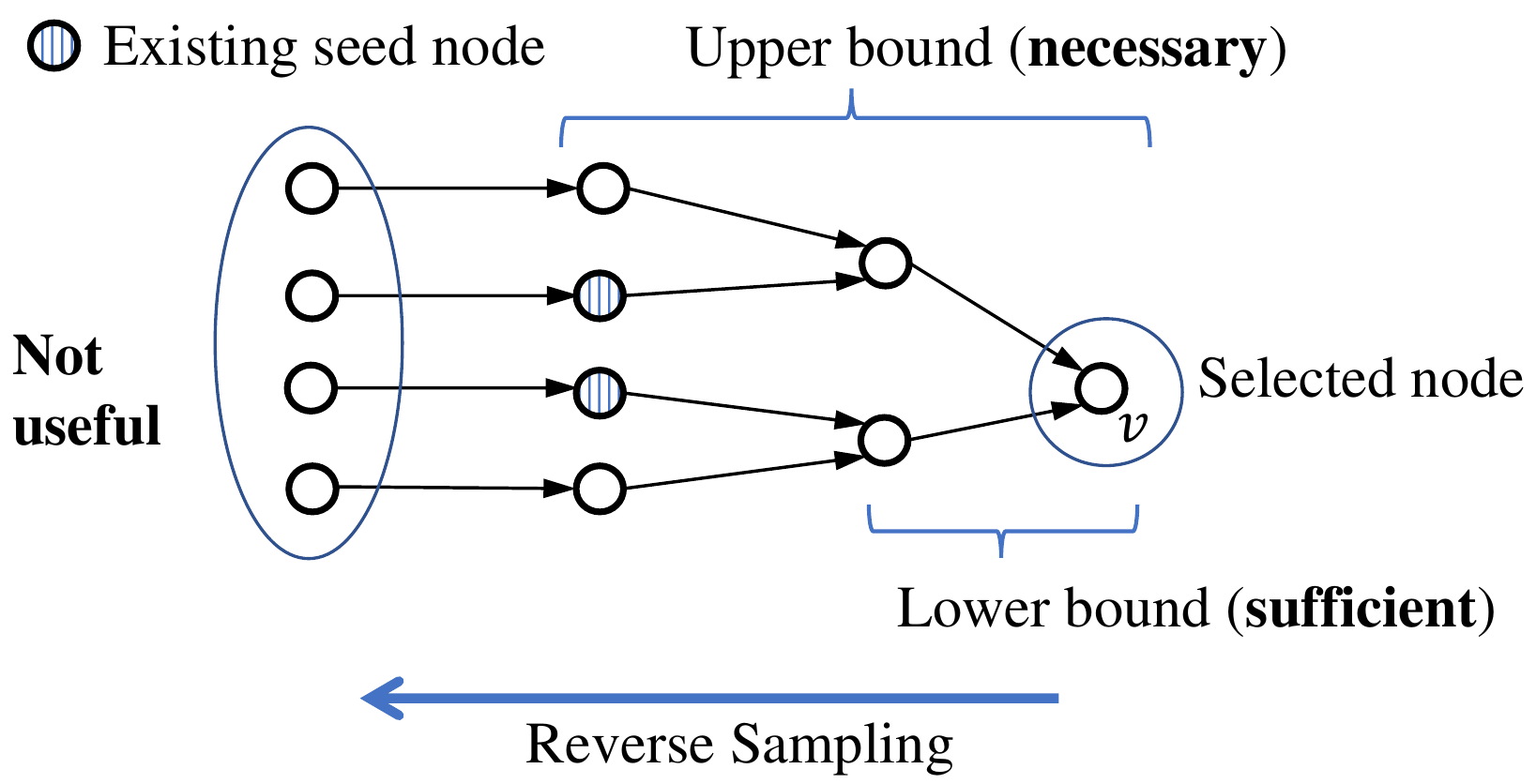} 
\end{center} 
\vspace{-2mm}
\caption{Illustration of the upper and lower bounds.}
\label{fig: reverse}
\end{figure}

Following Def. \ref{def: rr_tuple_v}, if we further let the target node $v$ be selected randomly, we have the complete sampling method shown in Alg. \ref{alg: rrtuple} where we first select a node $v$ from $V$ uniformly at random and then generate an RR-tuple of $v$ by using Alg. \ref{alg: rrtuple_v}. We call the output of Alg. \ref{alg: rrtuple} as an \textit{RR-tuple}
denoted as $\R$. Given that $\R$ is nothing but a random RR-tuple of some node, the notations $g(S, \R)$, $\overline{g}(S, \R)$ and $\underline{g}(S, \R)$ are defined analogously.
In turns out that $\R$ can be used to construct a desired estimator of our objective function.

\begin{corollary}
\label{coro: key}
For each $S \subseteq V$, we have
\begin{equation*}
\E[n\cdot \underline{g}(S ,\R)] \leq \E[n\cdot g(S ,\R)]=f(S)\leq E[n\cdot \overline{g}(S, \R)].
\end{equation*}
\end{corollary}
\begin{proof}
See Appendix \ref{proof: coro_key}
\end{proof}

It is worthy to note that the above bounds are constructed without using the activation functions, which is a key to making our framework work for an arbitrary activation function. Before presenting the algorithm for MIM, we state the following result showing the running time $\TIME_{\R}$ of generating one RR-tuple.

\begin{lemma}
\label{lemma: time}
$\E[\TIME_{\R}]=O(m \cdot \max_{|S|=1}\E[\underline{g}(S ,\R)]).$
\end{lemma}
\begin{proof}
Let $\E[\TIME_{\R_v}]$ be the expected running time of generating one RR-tuple of $v$ for a node $v \in V$, and we therefore have $\E[\TIME_{\R}]=\frac{1}{n}\sum_{v \in V}\E[\TIME_{\R_v}]$. According to Alg. \ref{alg: rrtuple_v}, it is clear that the running time of generating one RR-tuple of $v$ is bounded by the number the edges tested in Line 14 in Alg. \ref{alg: rrtuple_v}. For an edge $(u_1,u_2)$, it will be tested in Alg. \ref{alg: rrtuple_v} if and only if $u_2$ is in $\underline{S}_v$, and thus we have $\E[\TIME_{\R_v}]=\sum_{ (u_1,u_2)\in E}\E[\underline{g}(\{u_2\},\R_v)]$. As a result, 
\begin{align*}
\E[\TIME_{\R}]=\frac{1}{n}\sum_{v \in V}\sum_{ (u_1,u_2)\in E}\E[\underline{g}(\{u_2\},\R_v)]\\
=\sum_{ (u_1,u_2)\in E}\frac{1}{n}\sum_{v \in V}\E[\underline{g}(\{u_2\},\R_v)]
\leq m\cdot \max_{|S|=1}\E[\underline{g}(S ,\R)].
\end{align*}
\end{proof}

\subsection{Maximizing the Estimator}
According to Corollary \ref{coro: key}, the sample mean of $g(S ,\R)$ is an unbiased estimate of $\frac{f(S)}{n}$. Therefore, supposing that a collection $\Psi_l$ of $l \in \mathbb{Z}^+$ RR-tuples were generated, the error
\[|n\cdot \frac{\sum_{\R \in \Psi_l}g(S, \R)}{l}-{f(S)}|\] 
could be arbitrarily small for each $S \subseteq V$ provided that $l$ was sufficiently large, and consequently, the $S$ that can well maximize $\frac{n \cdot \sum_{\R \in \Psi_l}g(S, \R)}{l}$ should intuitively be a good solution to maximizing $f(S)$. Following this idea, let us define 
\begin{equation}
\label{eq: G}
\G(S,\Psi_l) \define \frac{n\cdot \sum_{\R \in \Psi_l}g(S, \R)}{l}
\end{equation}
and consider the following problem.

\begin{problem}
\label{problem: max_g}
Given a collection $\Psi_l$ of RR-tuples, compute 
\[\argmax_{\{* \subseteq V_c, |S|= k\}}\G(S,\Psi_l).\]
\end{problem}

\begin{algorithm}[t]
\caption{Greedy Algorithm}
\label{alg: greedy}
\begin{algorithmic}[1]
\State \textbf{Input:} a function $h$, a candidate set $V_c$ and $k \in \mathbb{Z}^+$.
\State \textbf{Output:} a subset of $V_c$.
\State $S^* \leftarrow \emptyset$;
\While {$|S^*| < k$}
\State $v^* \leftarrow \argmax_{v \in V} h(S^*\cup \{v\})-h(S^*)$;
\If	{$h(S^*\cup \{v\})-h(S^*)<0$} 
\Return $S^*$;
\Else  $~S^* \leftarrow S^* \cup \{v\}$;
\EndIf
\EndWhile
\Return $S^*$;
\end{algorithmic}
\end{algorithm}

A brief examination on Problem \ref{problem: max_g} shows that this problem suffers the similar inapproximability as Max-MIM, which however is not surprising as $\G(S,\Psi_l)$ can be made arbitrarily close to $f(S)$.

\begin{lemma}
\label{lemma: max_g}
DkS can be approximated within a factor of $2 \cdot \alpha$ if Problem \ref{problem: max_g} can be approximated within a factor of $\alpha$ for each $\alpha>1$.
\end{lemma}

\begin{proof}
See Appendix \ref{proof: lemma_max_g}.
\end{proof}

While $\G(S,\Psi_l)$ is hard to maximize directly, Corollary \ref{coro: key} suggests that the sandwich optimization technique is particularly useful because there exists upper and lower bounds which can be easily approximated.
For such purposes, let us define 
\begin{equation}
\label{eq: over_G}
\overline{\G}(S,\Psi_l) \define \frac{n\cdot \sum_{\R \in \Psi_l}\overline{g}(S, \R)}{l}
\end{equation}
and  
\begin{equation}
\label{eq: under_G}
\underline{\G}(S, \Psi_l) \define \frac{n\cdot \sum_{\R \in \Psi_l}\underline{g}(S, \R)}{l}.
\end{equation}
One can see that both $\overline{\G}(S,\Psi_l)$ and $\underline{\G}(S,\Psi_l)$ can be  maximized within a constant ratio as they are exactly the maximum coverage problem. In particular, by the celebrated result given by Feige \cite{feige1998threshold}, they can be approximated within a factor of $1-1/e$ by the greedy framework shown in Alg. \ref{alg: greedy}. Utilizing these results, the sandwich algorithm given in Alg. \ref{alg: sandwich} provides a solution to maximizing $\G(S,\Psi_l)$. In Alg. \ref{alg: sandwich}, given the upper and lower bounds, we run the greedy algorithm for each of them and return the one with the maximum objective value. The solution $S^*$ produced by Alg. \ref{alg: sandwich} has the following performance guarantee due to Lu \textit{et al.} \cite{lu2015competition}.

\begin{algorithm}[t]
\caption{Sandwich Approximation Strategy}\label{alg: sandwich}
\begin{algorithmic}[1]
\State \textbf{Input:} $\overline{\G}(S,\Psi_l),\G(S,\Psi_l),  \underline{\G}(S,\Psi_l)$, a candidate set $V_c$ and $k \in \mathbb{Z}^+$;
\State \textbf{Output:} a subset of $V_c$.
\State $\overline{S} \leftarrow$ Alg. \ref{alg: greedy} with ($\overline{\G}(S,\Psi_l),V_c, k$);
\State $\underline{S} \leftarrow$ Alg. \ref{alg: greedy} with ($\underline{\G}(S,\Psi_l),V_c, k$);
\State return $S^{*}=\argmax_{S \in \{\overline{S}, \underline{S}\}}{\G}(S,\Psi_l)$;
\end{algorithmic}
\end{algorithm}

\begin{lemma}[\cite{lu2015competition}]
\label{lemma: sandwich}
${\G}(S^*,\Psi_l) \geq \gamma(\Psi_l)\cdot {\G}(S_{opt}^{\Psi_l},\Psi_l)$, where $\gamma(\Psi_l)=\max \{\frac{{\G}(\overline{S},\Psi_l)}{\overline{\G}(\overline{S},\Psi_l)}, \frac{\underline{\G}(S_{opt}^{\Psi_l},\Psi_l)}{{\G}(S_{opt}^{\Psi_l},\Psi_l)}\}\cdot (1-1/e)$, $S_{opt}^{\Psi_l}$ is the optimal solution to Problem \ref{problem: max_g}, and $\overline{S}$ (resp., $\underline{S}$) is produced in Line 3 (resp., Line 4) of Alg. \ref{alg: sandwich}.
\end{lemma}

\begin{remark}
Here we have made a slight modification to the original sandwich strategy as we do not run the greedy algorithm for the objective function $\G(S,\Phi_l)$ itself. The reason for doing such is that the lazy-forward evaluation \cite{leskovec2007cost} method cannot be used for $\G(S,\Phi_l)$ as it is not submodular. Such an alteration does not hurt the performance ratio but will significantly improve the efficiency as discussed later in Sec. \ref{sec: discuss} as well as in experiments. 
\end{remark}

Following the above ideas, the proposed framework is shown in Alg. \ref{alg: framework}. Given a threshold $l$, we first generate $l$ random RR-tuples by using Alg. \ref{alg: rrtuple} and then apply the sandwich strategy. As we can see, the performance of this method is essentially bounded by the data-dependent ratio (i.e., $\gamma(\Psi_l)$) subject to the error due to sampling. Therefore, given $\epsilon_1, \epsilon_2 \in (0,1)$ and $N \in \mathbb{Z}^+$, our goal is to produce an $\big((1-\epsilon_2)\cdot \gamma(\Psi_l)-\epsilon_1 \big)$-approximation with probability at least $1-1/N$, where $\epsilon_1$ is the additive error and $\epsilon_2$ is the multiplicative error. Given the performance requirements, it remains to determine the parameter $l$, which will be discussed in the next subsection. 

Before moving to the next part, we herein mention one useful result concerning estimating the optimal value of $\E[n\cdot \underline{g}(S ,\R)]$. Let us use \[\underline{S}_{opt}=\argmax_{\{S \subseteq V_c, |S|= k\}} \E[\underline{g}(S, \R)]\] to denote the optimal solution to maximizing $\E[n\cdot \underline{g}(S ,\R)]$, and we have the following result. 

\begin{lemma}[\cite{tang2014influence}]
\label{lemma: estimate}
For each $\epsilon_0 >0$ and $N \in \mathbb{Z}^+$, there exists an algorithm which computes a real number $f_{\epsilon_0}^{\lo} \in \mathbb{R}$ such that 
\begin{equation}
\label{eq: lower_esti}
\E[n\cdot \underline{g}(\underline{S}_{opt}, \R)] \geq f_{\epsilon_0}^{\lo} \geq \frac{(1-1/e) \E[n\cdot \underline{g}(\underline{S}_{opt}, \R)]}{2(1+\epsilon_0)^2}
\end{equation}
holds with probability at least $1-1/N$, running in \[O(\frac{(m+n)(\ln N+k\ln n)(1+\epsilon_0)^3}{\epsilon_0^2}).\]
\end{lemma}
The proof follows from the existing works such as Theorem 2 in \cite{tang2014influence} or Theorem 2 in \cite{tong2017efficient}. Such an estimate $f_{\epsilon_0}^{\lo}$ will be used later in our algorithm.

\begin{algorithm}[t]
\caption{The Framework}\label{alg: framework}
\begin{algorithmic}[1]
\State \textbf{Input:} an instance $\M$ of MIM and $l \in \mathbb{Z}^+$.
\State \textbf{Output:} a subset of $V_c$.
\State Generate a collection $\Psi_l=\{\R_1,...,\R_l\}$ of $l$ RR-tuples by Alg. \ref{alg: rrtuple};
\State ${S}^* \leftarrow$ Alg. \ref{alg: sandwich}($\overline{\G}(S,\Psi_l),G(S,\Psi_l),  \underline{\G}(S,\Psi_l), V_c, k$);
\State Return ${S}^*$ and $\Psi_l$;
\end{algorithmic}
\end{algorithm}

\subsection{Parameter Setting}
We will first provide a sufficient condition to guarantee the desired performance, and then discuss how to make such a condition satisfied with a high probability.

\textbf{A Sufficient Condition.} Let $S^*$ and $\Psi_l$ be the output of Alg. \ref{alg: framework}, and $\epsilon_1$ and $\epsilon_2$ be the given parameters. We consider the following relations:
{ 
\begin{align}
\label{eq: c1}
&\G(S^*,\Psi_l)-f(S^*)\leq \epsilon_1 \cdot f(S_{opt})\\
\label{eq: c2}
&\G(S_{opt},\Psi_l)-f(S_{opt})\geq -\epsilon_2 \cdot f(S_{opt})\\
\label{eq: c3}
&\G(S^*,\Psi_l) \geq \gamma(\Psi_l) \cdot \G(S_{opt},\Psi_l)\\
&\epsilon_1, \epsilon_2 \in (0,1). \nonumber
\end{align}}
Eqs. (\ref{eq: c1})-(\ref{eq: c3}) would be sufficient to ensure the desired the performance as shown in the following lemma of which the proof is elementary. 
\begin{lemma}
\label{lemma: ratio}
We have $f(S^*) \geq \big((1-\epsilon_2)\cdot \gamma(\Psi_l)-\epsilon_1 \big)\cdot f(S_{opt})$ provided Eqs. (\ref{eq: c1}), (\ref{eq: c2}) and (\ref{eq: c3}).
\end{lemma}



\textbf{Determining the parameter $l$.} According to Theorem \ref{theorem: key_1}, Eqs. (\ref{eq: c1}) and (\ref{eq: c2}) can be satisfied when $l$ is sufficiently large. In particular, we have the following result showing the bound of the sufficient number of samples.

\begin{lemma}[\cite{report}]
\label{lemma: l}
Eqs. (\ref{eq: c1}) and (\ref{eq: c2}) simultaneously hold with probability at least $1-2/N$ if $l \geq  \frac{\max (l_1,l_2)}{f(S_{opt})}$ where
\begin{equation}
\label{eq: l_1}
l_1 \define \dfrac{n\cdot (\ln \binom{n}{k} + \ln N)(2+\epsilon_1)}{(\epsilon_1 )^{2}}
\end{equation}
and
\begin{equation}
\label{eq: l_2}
l_2 \define \frac{2\cdot n \cdot \ln N }{(\epsilon_2)^{2}}
\end{equation}
\end{lemma}
\begin{proof}
See Appendix \ref{proof: lemma_time}
\end{proof}

The above lemma does provide a sufficient bound but $f(S_{opt})$ is unknown to us. Such an issue can be overcome by computing a lower bound in advance. In particular, the $f_{\epsilon_0}^{\lo}$ mentioned earlier would perfectly meet our purpose. 
\begin{lemma}
\label{lemma: lower_opt}
It follows immediately from Theorem \ref{theorem: key_1} and Lemma \ref{lemma: estimate} that 
\begin{align*}
f_{\epsilon_0}^{\lo}&\leq \E[n\cdot \underline{g}(\underline{S}_{opt}, \R)] \\
&\leq \E[n\cdot {g}(\underline{S}_{opt}, \R)]=f(\underline{S}_{opt}) \leq f(S_{opt}).
\end{align*} 
\end{lemma}

\begin{remark}
Different from the existing works, we compute the lower bound by analyzing $\E[\underline{g}(S, \R)]$ instead of the original objective function $f$, because $\G(S,\Psi_l)$ itself is hard to maximize. Consequently, the obtained bound is looser, leading to an increase in time complexity. However, according Lemma \ref{lemma: time}, the time complexity of generating one RR-tuple is bounded by $O(m \cdot \max_{|S|=1}\E[\underline{g}(S ,\R)])$ which concerns only the lower bound function. Therefore, $f_{\epsilon_0}^{\lo}$ is sufficient to ensure that the running time of our algorithm is near-optimal, as shown later in this section. In addition, we can see from Eqs. (\ref{eq: c1}) and (\ref{eq: c2}) that among the three estimators $\G$, $\overline{\G}$ and $\underline{\G}$, we only need to bound $\G$ to ensure the data-dependent approximation ratio.
\end{remark}

\begin{algorithm}[t]
\caption{Reverse Sandwich Algorithm}
\label{alg: rs}
\begin{algorithmic}[1]
\State \textbf{Input:} an instance $\M$ of MIM, $\epsilon_0, \epsilon_1,\epsilon_2 \in (0,1), N>1$;
\State \textbf{Output:} a node set $S_{RS} \subseteq V_c$;
\State Compute $f_{\epsilon_0}^{\lo}$ by Lemma \ref{lemma: estimate} with $\epsilon_0$;
\State $l \leftarrow \max(l_1, l_2)/f_{\epsilon_0}^{\lo}$ according to Lemma \ref{lemma: l};
\State $(S_{RS},\Psi_l) \leftarrow$ Alg. \ref{alg: framework} with $(\M, l)$;
\State \textbf{Return} $S_{RS}$ and $\gamma(\Psi_l)$;
\end{algorithmic}
\end{algorithm}

\subsection{Reverse Sandwich (RS) Framework}
\label{subsec: rs}
Putting the above together, we have the complete algorithm shown in Alg. \ref{alg: rs}, termed as \textit{Reverse Sandwich (RS) algorithm}. Given an instance $\M$ of the MIM problem and the parameters $\epsilon_0$, $\epsilon_1, \epsilon_2$ and $N$, the algorithm consists of the following steps:
\begin{enumerate}
\item  Compute the estimate $f_{\epsilon_0}^{\lo}$ according to Lemma \ref{lemma: estimate}. Determine $l$ by $l=\frac{\max(l_1, l_2)}{f_{\epsilon_0}^{\lo}}$ according to Lemmas \ref{lemma: l}.
\item Call the framework Alg. \ref{alg: framework} with $\M$ and $l$ to obtain the final solution $S_{RS}$ together with the used samples $\Psi_l$.
\end{enumerate}

\textbf{\textbf{Running Time and Implementation}.} Due to Lemma \ref{lemma: estimate}, Step (1) runs in \[O(\frac{(m+n)(\ln N+k \ln n)(1+\epsilon_0)^3}{\epsilon_0^2}).\] For Line 3 in Alg. \ref{alg: framework}, by Lemma \ref{lemma: estimate} and \ref{lemma: l}, the total number of RR-tuples is $O(\frac{n\cdot (\ln N+k\ln n)}{(\min (\epsilon_1,\epsilon_2))^2\cdot f_{\epsilon_0}^{\lo}})$, and combining Lemma \ref{lemma: time}, its total running time is 
\begin{equation}
\label{eq: time}
O(\frac{(m+n)(\ln N+k\ln n)(1+\epsilon_0)^2}{(\min (\epsilon_1,\epsilon_2))^2}).
\end{equation}
Line 4 in Alg. \ref{alg: framework} runs in linear to the input size which is bounded by the time of generating the RR-tuples \cite{vazirani2013approximation}. Therefore, the running time of Step (2) is the same as Eq. (\ref{eq: time}). Suppose $\epsilon_1$ and $\epsilon_2$ are given as the required accuracy, to minimize the total running time of Step (1) and Step (2), we ideally seek the $\epsilon_0$ that can minimize $\frac{K\cdot (1+\epsilon_0)^3}{(\epsilon_0)^2}+\frac{(1+\epsilon_0)^2}{(\min (\epsilon_1,\epsilon_2))^2}$, where $K$ is determined by the hidden constants behind the asymptotically bounds.  Since it is not possible to compute $K$, we follow the idea in Tang \textit{et al.} \cite{tang2015influence} and set $\epsilon_0=K \cdot \epsilon_1$ where in experiments we simply adopt $K=100$. For simplicity, we may consider a unified error $\epsilon$ and adopt the setting that $\epsilon_1=\epsilon_2=\epsilon$ and $\epsilon_0=K \cdot \epsilon_1$. Under such a setting, the running time is $O(\frac{(m+n)(\ln N+k\ln n)}{\epsilon^2})$ and the performance ratio is $(1-\epsilon)\cdot \gamma(\Psi_l)-\epsilon$. The results are summarized as follows.

\begin{theorem}
\label{theorem: main}
With probability at least $1-1/N$,  RS algorithm produces a solution $S_{RS}$ such that
\[f(S_{RS})\geq \big((1-\epsilon)\cdot \gamma(\Psi_l)-\epsilon \big)\cdot f(S_{opt}),\]
running in $O(\frac{(m+n)(\ln N+k\ln n)}{\epsilon^2})$ in expect, where $\Psi_l$ is the collection of samples used during the process. 
\end{theorem}
\begin{proof}
See Appendix \ref{proof: theorem_main}.
\end{proof}

\begin{remark}
\label{remark: ratio}
In terms of $m$ and $n$, the running time of Alg. \ref{alg: rs} is near-linear up to a factor of $\ln n$, which matches the near-optimal one for the IM problem. In another issue, $\gamma(\Psi_l)$ is bounded by $\frac{{\G}(\overline{S},\Psi_l)}{\overline{\G}(\overline{S},\Psi_l)}\cdot (1-1/e)$ which can be easily computed after each run of Alg. \ref{alg: rs}, which provides a method to evaluate the approximation ratio in real-time. 
\end{remark}

\section{Further Discussions}
\label{sec: discuss}

\subsection{The Naive Greedy Algorithm}
\label{subsec: nrg_algorithm}
Even though the greedy algorithm does not provide any performance bound for the MIM problem, one might still believe it is practically effective. The most straightforward way to apply the greedy algorithm is to leverage the Monte-Carlo simulation to estimate the objective function, which is however very time-consuming. Moving forward slightly, the analysis in Sec. \ref{subsec: sampling} shows that the reverse sampling method provides an unbiased estimate (i.e., $\G(S,\Psi_l)$) of the objective function of MIM, and therefore, a greedy algorithm can be obtained by first generating a bunch of samples $\Psi_l$ and then apply Alg. \ref{alg: greedy} to $\G(S,\Psi_l)$, which is somehow the best algorithm that can be trivially inferred from the existing techniques. We denote such an algorithm as NR-Greedy and will take it as a baseline method in experiments. Suppose $l$ RR-tuples were used in NR-Greedy, the running time would be $O(knl\cdot \TIME_R)$ so setting $l$ as the same as that in RS results in a running time of $O(knm)$ which is not scalable. For such a heuristic, we set the number of used RR-tuples according to Eq. (\ref{eq: l_2}). However, NR-Greedy is still time-consuming under such a setting, as shown later in experiments.

\subsection{Enabling Data-driven Models}
As one may have noticed, the framework in Sec. \ref{sec: framework} utilizes the upper and lower bounds which do not rely on any specific form of the activation function. 
The rationale behind it is that the lower and upper bounds correspond to the sufficient and necessary conditions which are derived by only taking account of the arrival time of the cascades. As a result, the RS algorithm can be taken as a framework which applies to a wide range of activation functions. While activation functions in a simple form cannot capture the real cases, one can utilize data-driven approaches to learn the activation function by leveraging historical data of user profile or cascade characteristics. Data-driven activation functions would make the IMC model more expressive while keeping the combinatorial structures used in our algorithm in Sec. \ref{alg: framework}.

\subsection{On the Upper and Lower Bounds}
Recall that the bounds used in our algorithm are $\E[n\cdot \underline{g}(S ,\R)]$ and $E[n\cdot \overline{g}(S, \R)]$. We now provide sufficient conditions for those bounds to be tight. 

\begin{definition}[$c_{new}$-dominating Activation Function]
\label{def: dominating}
$F_v$ is \textit{$c_{new}$-dominating} if $F_v(O)=c_{new}$ for each $O \subseteq V \times C$ where $c_{new} \in \{c\}_{(v,c) \in O}$.
\end{definition}

Informally speaking, $c_{new}$ always wins the competition during the diffusion if $F_v$ is $c_{new}$-dominating at each node. In such a case, the upper bound is tight.

\begin{lemma}
\label{lemma: exp_upper}
If $F_v$ is $c_{new}$-dominating for each $v \in V$, we have $\E[n\cdot \overline{g}(S ,\R)]=f(S)$ for each $S \subseteq V$.
\end{lemma}

\begin{proof}
See Appendix \ref{proof: lemma: exp_upper}.
\end{proof}

Similarly, we could consider the case where $c_{new}$ is always dominated by other cascades whenever competition appears, and thus we have the following pair.

\begin{definition}[$c_{new}$-dominated Activation Function]
\label{def: dominated}
$F_v$ is $c_{new}$-dominated if $F_v(O)\neq c_{new}$ for each $O\subseteq V \times C$ where $|\{c\}_{(v,c) \in O}|\geq 2$ and $c_{new} \in \{c\}_{(v,c) \in O}$.
\end{definition}

\begin{lemma}
\label{lemma: exp_lower}
If $F_v$ is $c_{new}$-dominated for each $v \in V$, we have $\E[n\cdot \underline{g}(S ,\R)]=f(S)$ for each $S \subseteq V$.
\end{lemma}

\begin{remark}
These results indicate that the bounds derived in Sec. \ref{sec: framework} are not only technically useful but also have intuitive meanings: in terms of the influence of $c_{new}$, they correspond to the most optimistic scenario and most pessimistic scenario, respectively. This again echoes the observation in Sec. \ref{sec: framework} that the lower and upper bounds can be derived without knowing the activation function. When there is no existing cascade, we can see that the upper and lower bounds are identical to the objective function and the whole framework in Sec. \ref{sec: framework} reduces to the one given in \cite{tang2014influence} for IM, which can also be seen from the fact that $\overline{S}_v$ and $\underline{S}_v$ would always be the same if there is no existing cascade. From this perspective, the performance of the proposed algorithm is affected by the degree of competition, which will be examined via experiments. In another issue, Defs. \ref{def: dominating} and \ref{def: dominated} suggest that carefully manipulating the activation function may provide tighter upper and lower bounds. Finally, the objective function with activation functions in Def. \ref{def: dominating} and \ref{def: dominated} meet the upper or lower bound, and therefore they are submodular.
\end{remark}

\section{Experiments}
\label{sec: exp}

Since no solution provides a worst-case performance guarantee for the MIM problem, practical evaluations are important for judging the merits of the algorithms. In this section, we present the experimental results. 

\begin{table}[t]
\centering
\caption{Datasets}
\label{table: data}
\begin{tabular}{@{}lllll@{}}
\toprule
Name & Hepph & Higgs &Youtube& Livejournal\\
\midrule
\# of Nodes & 34546 & 100,000& 1,134,890& 3,997,962\\
\# of Edges &421,578 & 192,972& 2,987,624& 34,681,189\\
\bottomrule
\end{tabular}
\end{table}


\textbf{Dataset.} We select four representative datasets, \textit{Hepph, Higgs, Youtube} and \textit{Livejournal}, which have various network sizes and have been widely adopted for studying influence maximization. Hepph is a citation network of high energy physics on Arxiv.org. Youtube and Livejournal are graphs of popular online social networks borrowed from SNAP \cite{leskovec2015snap}. Higgs consists of the user actions regarding the discovery of a new particle with the features of the elusive Higgs boson collected from Twitter \cite{de2013anatomy}. A brief summary of the statistics is given in Table \ref{table: data}. 

\textbf{Cascade Setting.} We consider four different methods for setting the prorogation probability. The prorogation probability in Hepph is uniformly set as $0.1$. In Higgs, the prorogation probability on each edge $(u,v)$ is set to be proportional to the frequency of actions between $u$ and $v$ \cite{tong2018misinformation}.  For Youtube, the probability on each edge is sampled from the exponential distribution transferred to $[0,1]$ with mean $0.01$.  For Livejournal, we adopt the weighted-cascade setting where $p_{(u,v)}=1/|deg_v^-|$ and $deg_v^-$ is the in-degree of $v$. For each dataset, we deploy four existing cascades. The size of the seed set of each existing cascades is $1\%$ of the total number nodes unless otherwise specified. The seed set size of $c_{new}$ is enumerated from $20$ to $50$. 


\textbf{Activation Function.} We consider three types of activations. The first two activation functions are those defined in Defs. \ref{def: cascade_function} and \ref{def: neighbor_function}, where the priorities are randomly generated in advance and fixed in experiments. The last one is a non-deterministic activation function where each coming cascade has the same probability to activate the target node. We denote these settings as Cascade-based Activation (\textit{CA}), Neighbor-based Activation (\textit{NA}) and Random Activation (\textit{RA}), respectively. 

\textbf{Algorithms.} Besides the RS algorithm proposed in this paper, we consider the NR-Greedy discussed in Sec. \ref{subsec: nrg_algorithm} and the MaxInf algorithm which ignores the existing cascades and selects the nodes that can maximize the influence of $c_{new}$ as the seed nodes. 

For each set of the experiment, we record the resulted influence and running time of each algorithm as well as the data-dependent ratio discussed in Remark \ref{remark: ratio}. Our experiments were done on an Intel Xeon Platinum 8000 Series processor with sophisticated parallelizations for each algorithm. The source code will be made publicly available.

\subsection{Observations}
The figures and tables can be found in Appendix \ref{sec: results}.

\textbf{Overall Observations.} The main outcome is shown in Figs. \ref{fig: main_Hepph}, \ref{fig: main_Higgs}, \ref{fig: main_Youtube} and \ref{fig: main_Livejournal}, where the table shows the running time and the ratio $\frac{{\G}(\overline{S},\Psi_l)}{\overline{\G}(\overline{S},\Psi_l)}$ for $k=50$. It can be seen that RS outperforms other methods in most cases. In particular, RS achieves the highest influence and it terminates within several minutes even on large datasets. NR-Greedy is time-consuming and its performance is not comparable to that of RS. MaxInf is worse than RS by a moderate margin on Hepph and Livejournal but can be ineffective and even worse than NG-Greedy on Youtube and Higgs. For the ratio $\frac{{\G}(\overline{S},\Psi_l)}{\overline{\G}(\overline{S},\Psi_l)}$, one can see that it is consistently larger than 0.7 and extremely close to 1 on large datasets, which implies that the data-dependent ratio in practice is not far from $1-1/e$ which is the theoretically optimal. Such results implicitly indicate that the submodularity is not dramatically altered even though it does not hold in MIM. Overall, RS enjoys high practical performance, though none of these three methods theoretically dominates the others.\footnote{For example, when all the nodes have been already selected as seed nodes of existing cascades, NR-Greedy still follows the hill-climbing manner but RS is equivalently to randomly select $k$ nodes; when there is no existing cascade, the three methods are essentially identical.} 

\textbf{Hyper-parameter Study of RS.} As mentioned earlier, there are three parameters $N, K$ and $\epsilon$ in RS. Our default setting is $(N,K,\epsilon)=(10000,100,0.3)$, which works smoothly in our experiments. We examine the impact of these three parameters by changing one of them with the other two being fixed as the default values. The results on Youtube are given in Table \ref{table: hyper_youtube}. Overall, the resulted influence does not vary much, except for $\epsilon=0.9$ under which the influence is slightly less than others. There is a clear trade-off between $\epsilon$ and the running time, and we have the similar observation regarding $K$. The results here suggest that setting an extremely small $\epsilon$ or $K$ would make RS unnecessarily time-consuming. For parameter $N$, the results are stable in terms of both the effectiveness and efficiency. On Livejournal (Table \ref{table: hyper_livej}), we have the similar observations.

\textbf{NR-Greedy with More Samples.} Recall that the performance of NR-Greedy depends on the number of used RR-tuples. In order to further compare RS and NR-Greedy, we gradually increase the number of  RR-tuples used in NR-Greedy until its performance reaches that of RS, during which we record the running time. The results of this part on Hepph are given in Table \ref{table: exp2_hepph}. According to the table, NR-Greedy demands more than one hour to achieve an influence that can be produced by RS within one second, and under the RA setting, it has not caught up RS after running several hours.

\textbf{Increasing the Degree of Competition.} As observed in Figs. \ref{fig: hepph_cas}, MaxInf occasionally has the same performance as RS, and it can be effective especially when the competition between cascades is not fierce because in such cases MIM is close to the classic IM problem for which MaxInf is powerful. To demonstrate such intuitions, we enlarge the seed set of the existing cascades to increase the degree of competition and repeat the experiments. The size of seed set of existing cascades is increased to $5\%$  and $10\%$ of the total nodes, and the results on Hepph under CA are shown in Table \ref{table: competition} in which each cell show the ratio of the performance of MaxInf to that of RS. The results show an increased superiority of RS over MaxInf when stronger competitions arise. For example, under the setting $k=50$, the ratio reduces to 0.750 from 0.919 when size of seed set of existing cascades increases from $1\%$ to $10\%$.   The similar observation can be made on Livejournal, as shown in Fig. \ref{fig: rswithmaxinf_live}. We can see that the margin between the two curves becomes larger with the increase in seed set size of the existing cascades.

\textbf{Overlap between Seed Sets.} One interesting problem is to see if the seed nodes selected by different methods are totally different or almost the same. To this ends, we compute the percent of the seed nodes shared by different methods, and the results on Hepph are shown in Fig. \ref{fig: overlap}. We observe that Sandwich and MaxInf have about 75\% overlapping seed nodes, but NR-Greedy shares very few with other methods. We consistently have such an observation for all the settings, CA, NA and RA. One plausible reason is that Sandwich and MaxInf all utilize the IMM framework, while NR-Greedy is more like a naive greedy algorithm.

\textbf{Extreme Settings.} To further examine the scalability of the algorithms, we test Youtube and Livejournal under the setting $k=1,000$, and the results of this part can be found in Table \ref{table: extreme}. On Youtube, we observe that RS still runs smoothly, while MaxInf cannot handle large $k$ and it is more time-consuming compared to RS. The main drawback of MaxInf is that it is very memory-consuming and it works only for $k$ less than $150$. Because MaxInf ignores the existing cascades, the cascade $c_{new}$ can spread more widely in sampling and therefore the average size of each RR-tuple used in MaxInf is lager than that in RS. Thus, MaxInf would consume more memory. On Live, even though it is larger than Youtube, both MaxInf and RS could run using a moderate amount of memory, which suggests that the memory consumption is more sensitive to the probability setting than the graph size. They have the similar running time but RS could result in a higher influence. For NR-Greedy, its main problem is not memory issue but the running time taken in greedy node selection. It takes more than $20$ hours on Youtube when $k=1000$, while RS only takes $30$ minutes. We have further tested graphs with one billion nodes and four billion edges under the constant setting, and we found that none of the considered algorithms can run with $371$G memory taking less than 24 hours.

\section{Conclusion}
\label{sec: con}
In this paper, we study the multi-cascade influence maximization problem. We propose the independent multi-cascade model based on the activation function which characterizes the mutual effect between different cascades. We prove that the multi-cascade influence maximization problem is hard to approximate and design the RS algorithmic framework for node selection. The designed algorithm leverages upper and lower bounds which reveal the key structure of multi-cascade diffusion.  Demonstrated by experiments, RS outperforms baseline methods, and it is sufficiently scalable to handle large datasets.

\textbf{Future Work.} One assumption made in this paper is that the network work structure as well as the seed placement of existing campaigns is known in advance, which is however not realistic. Such an assumption could be relaxed by (a) assuming that we are given a snapshot of the diffusion process instead of the network structure and (b) inferring the edge probability and network structure using historical data. In addition, this paper assumes the activation function can be learned in advance while such a learning process is not trivial. Therefore, designing appropriate approaches for learning activation function is one promising direction of future work. Furthermore, different from the classic IM problem, the analysis of MIM for cascade models cannot be trivially extended to the threshold models, which is another direction of future work. 

\appendices
\section{Experimental Results}
\label{sec: results}

\begin{figure}[h!]
\centering
\subfloat[{[Hepph, CA]}]{\label{fig: hepph_cas}\includegraphics[width=0.23\textwidth]{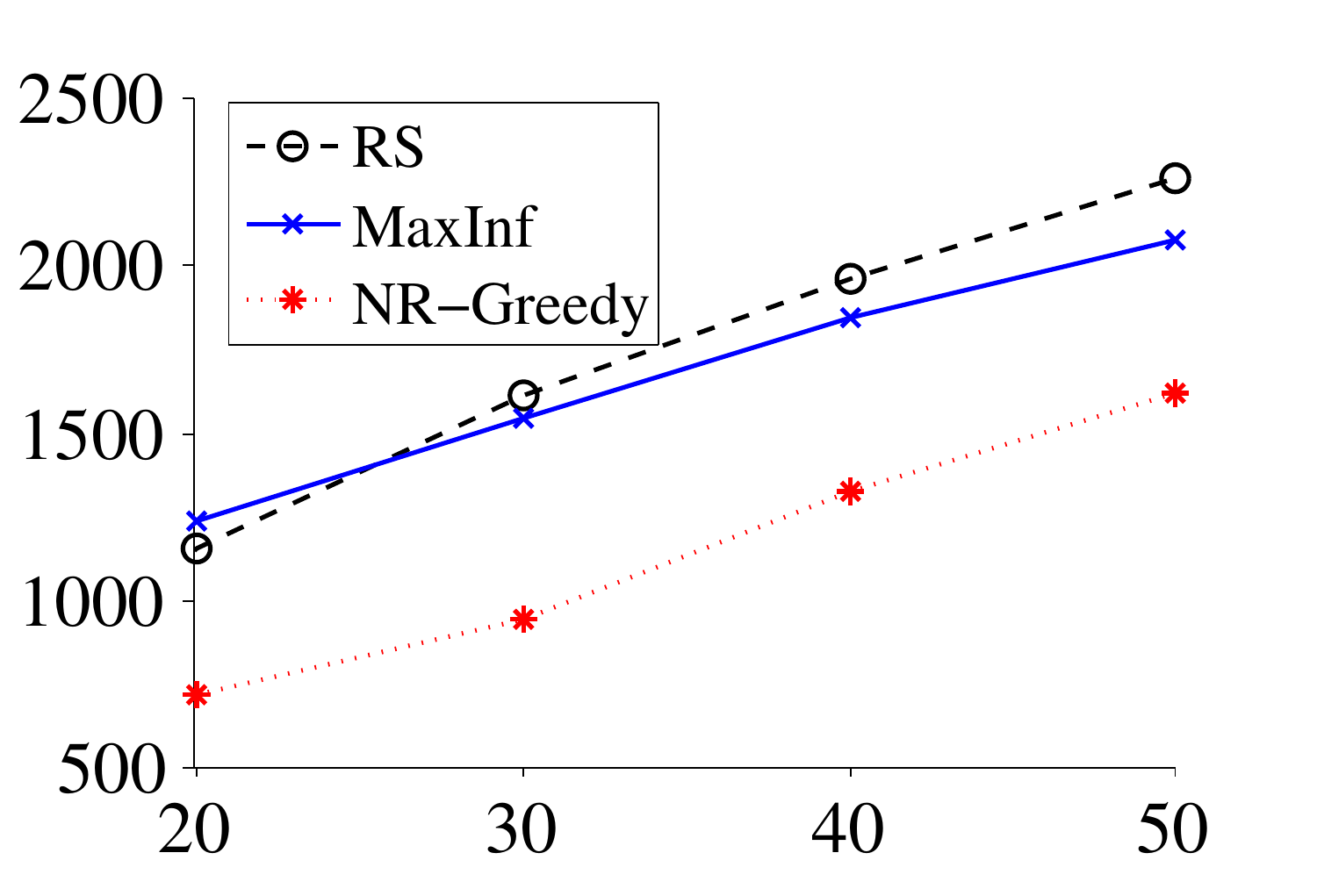}}
\subfloat[{[Hepph, NA]}]{\label{fig: hepph_nei}\includegraphics[width=0.23\textwidth]{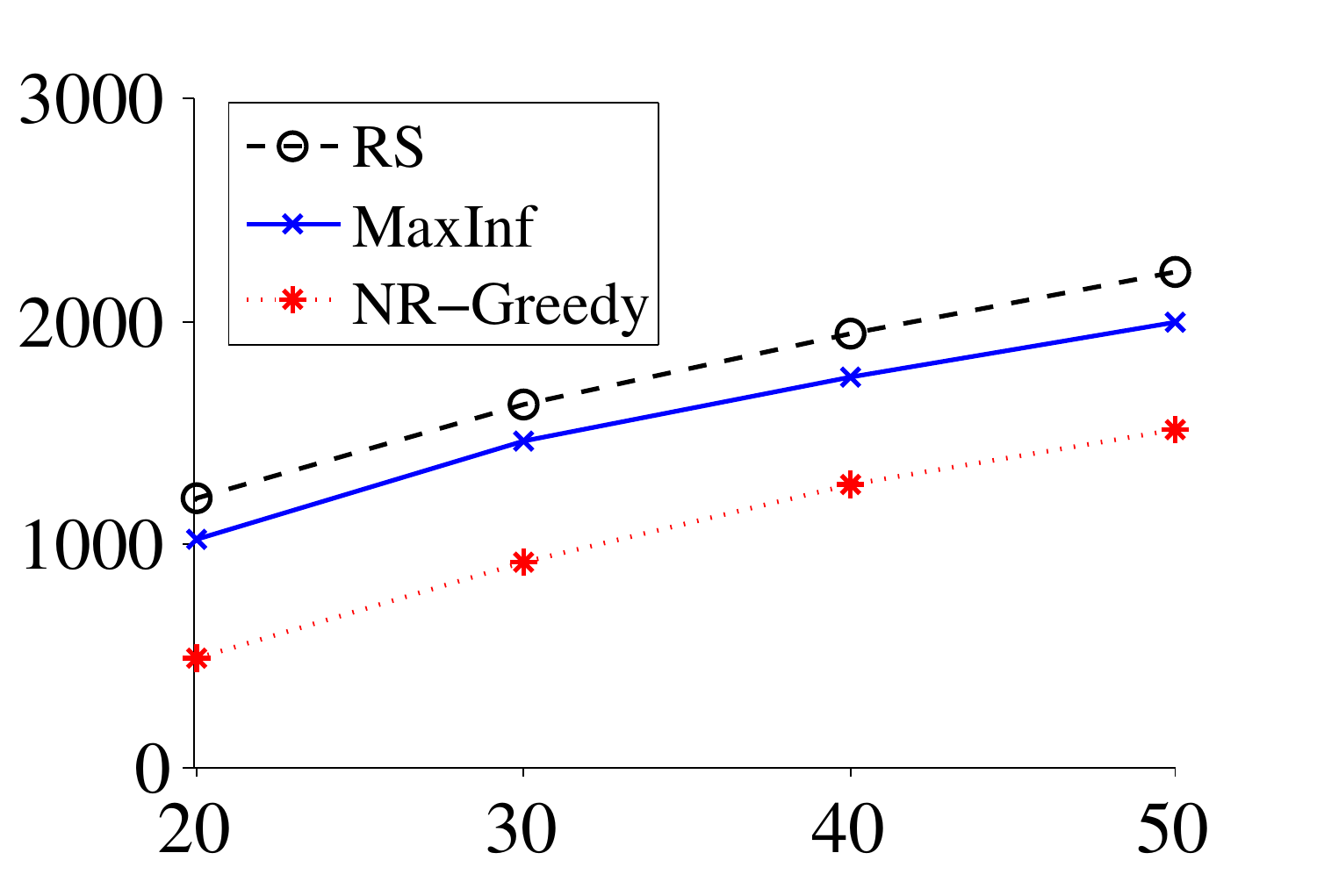}}

\subfloat[{[Hepph, RA]}]{\label{fig: hepph_ran}\includegraphics[width=0.23\textwidth]{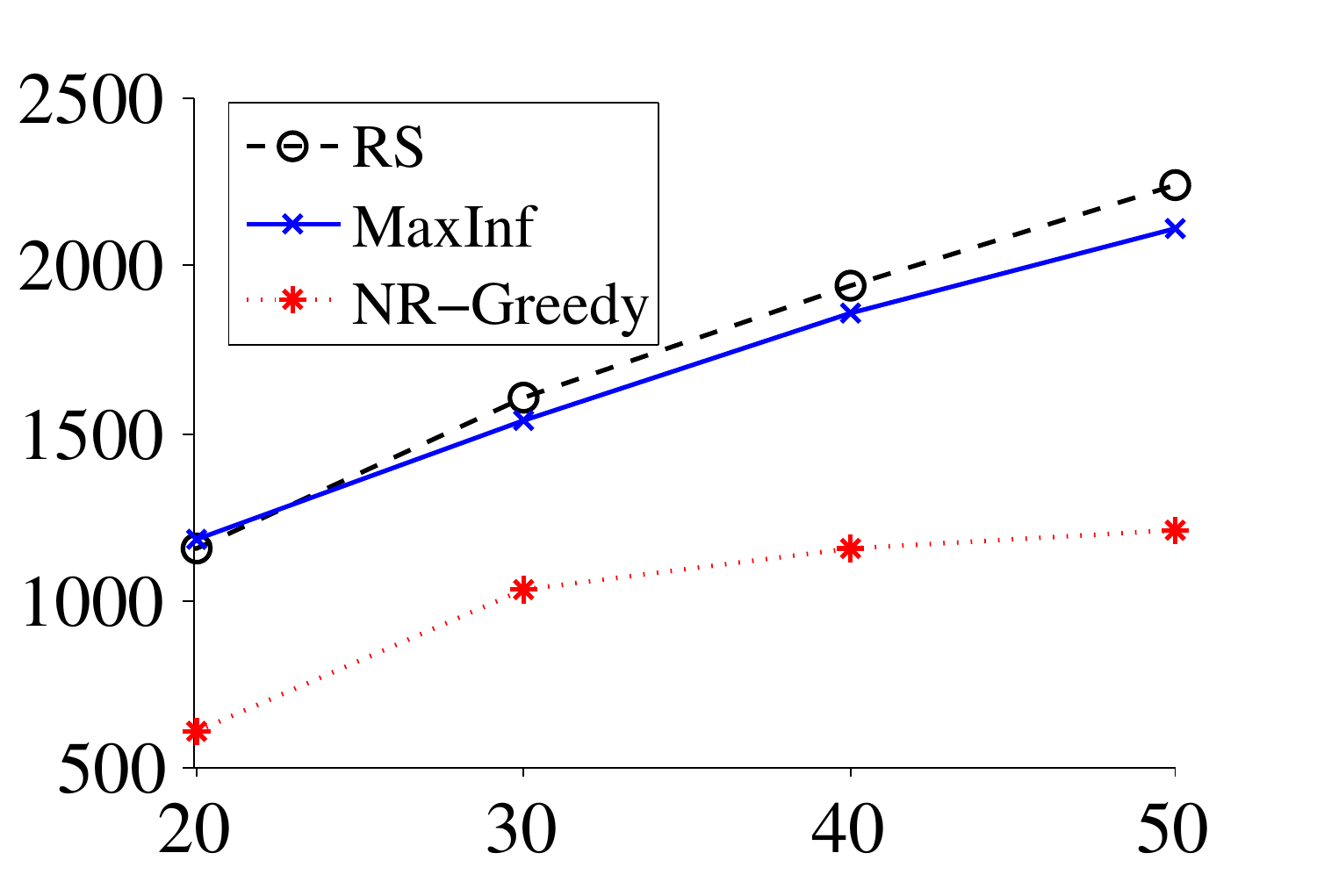}}
\raisebox{0.25\height}{\subfloat{\label{fig: hepph_ratio}\includegraphics[width=0.23\textwidth]{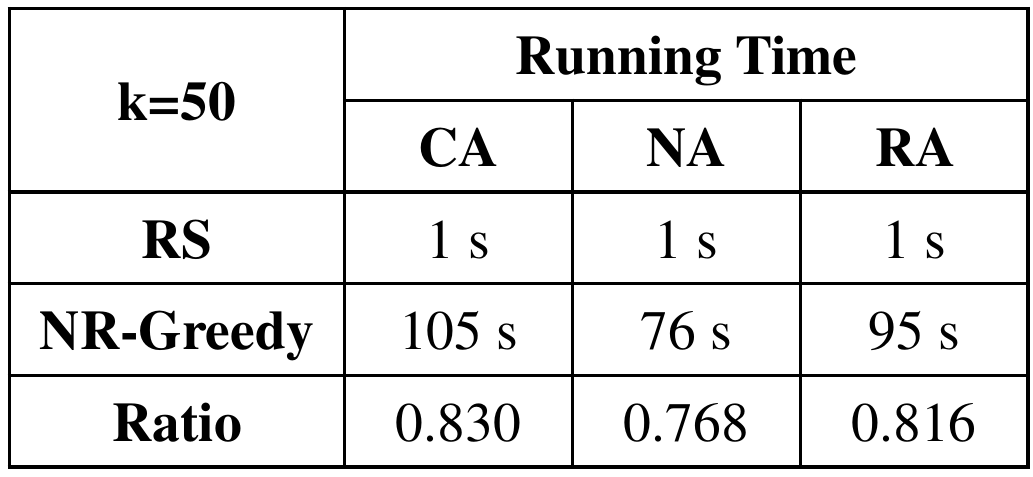}}}

\caption{Experimental Results on Hepph.}
\label{fig: main_Hepph}
\end{figure}

\begin{figure}[h!]
\centering
\subfloat[{[Higgs, CA]}]{\label{fig: higgs_cas}\includegraphics[width=0.23\textwidth]{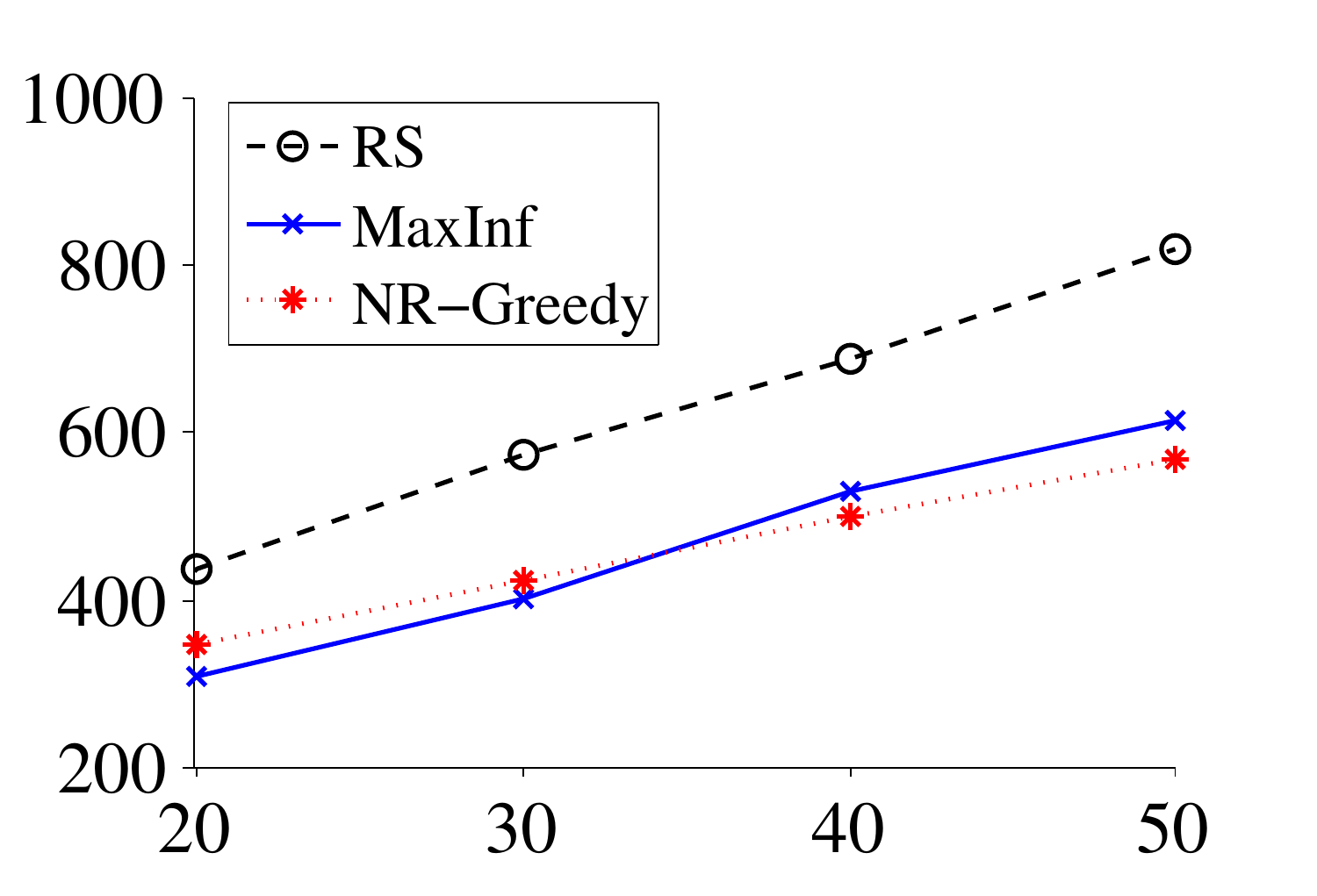}}
\subfloat[{[Higgs, NA]}]{\label{fig: higgs_nei}\includegraphics[width=0.23\textwidth]{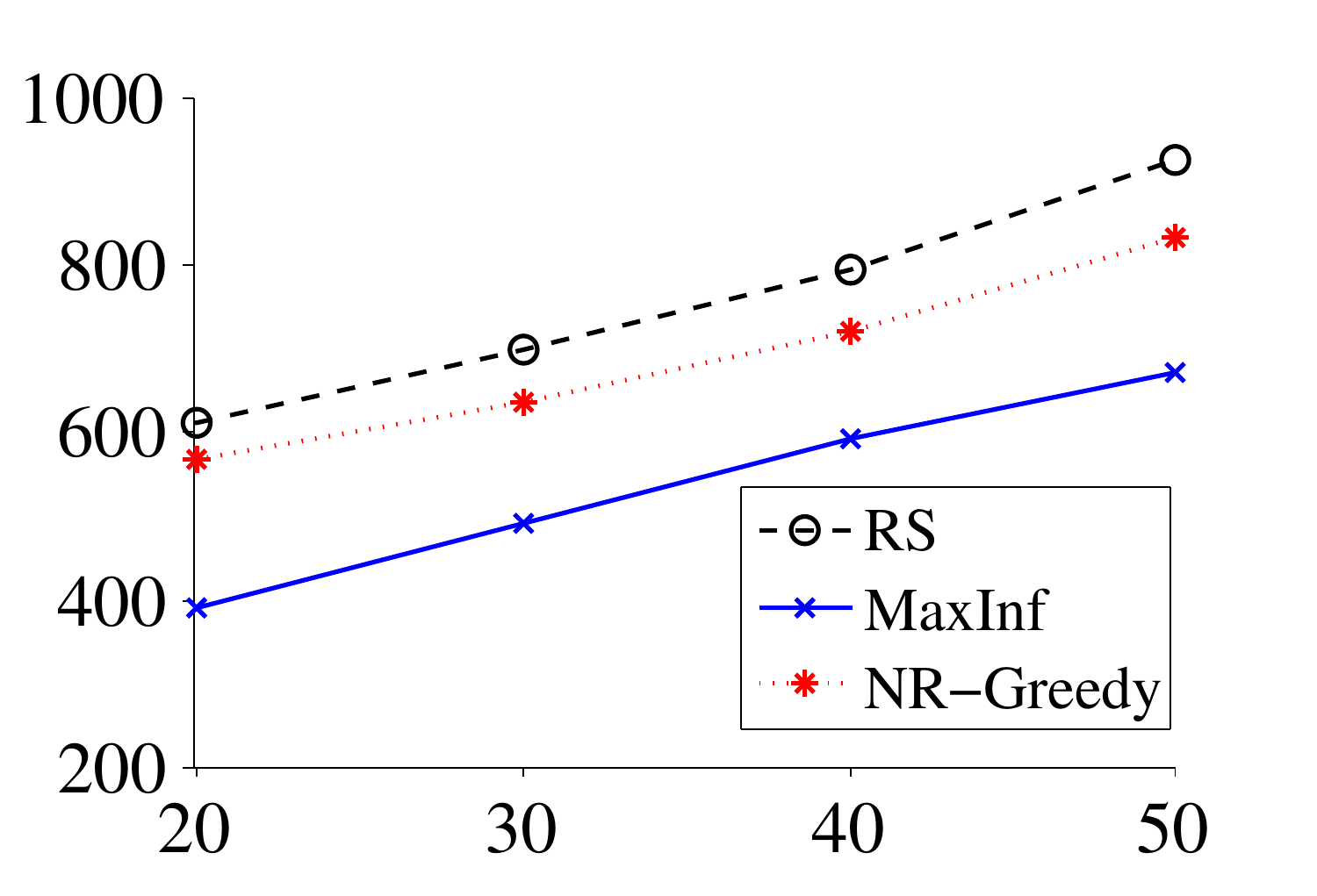}}

\subfloat[{[Higgs, RA]}]{\label{fig: higgs_ran}\includegraphics[width=0.23\textwidth]{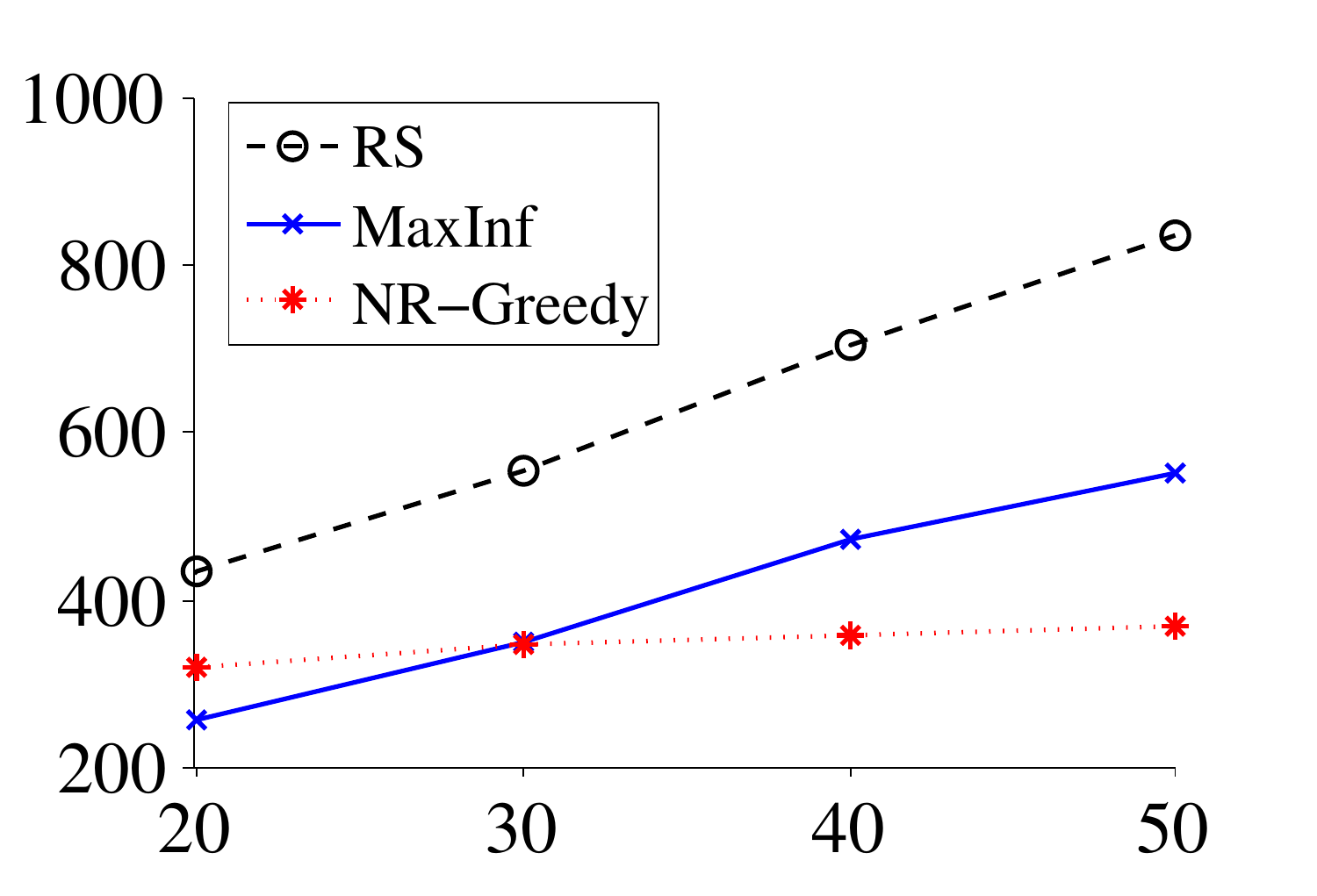}}
\raisebox{0.25\height}{\subfloat{\label{fig: higgs_ratio}\includegraphics[width=0.23\textwidth]{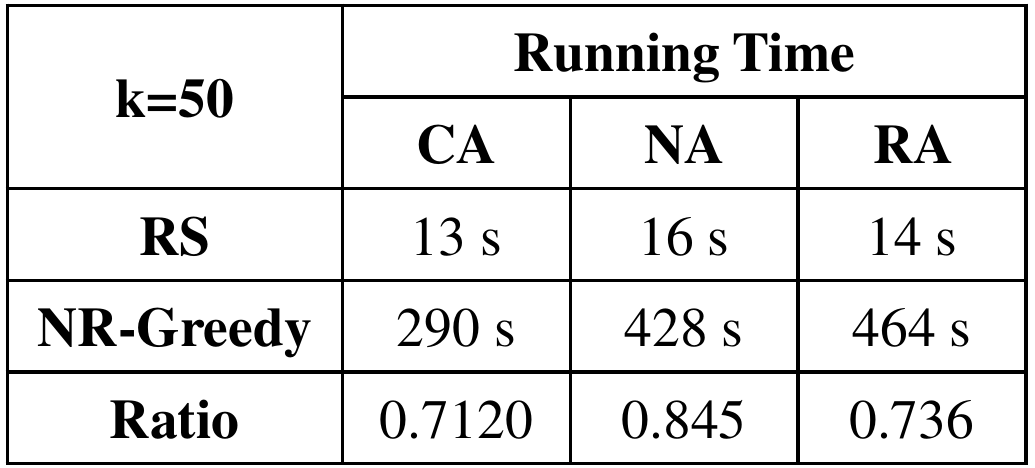}}}

\caption{Main Experimental Results on Higgs.}
\label{fig: main_Higgs}
\end{figure}

\begin{figure}[h!]
\centering
\subfloat[{[Youtube, CA]}]{\label{fig: youtube_cas}\includegraphics[width=0.23\textwidth]{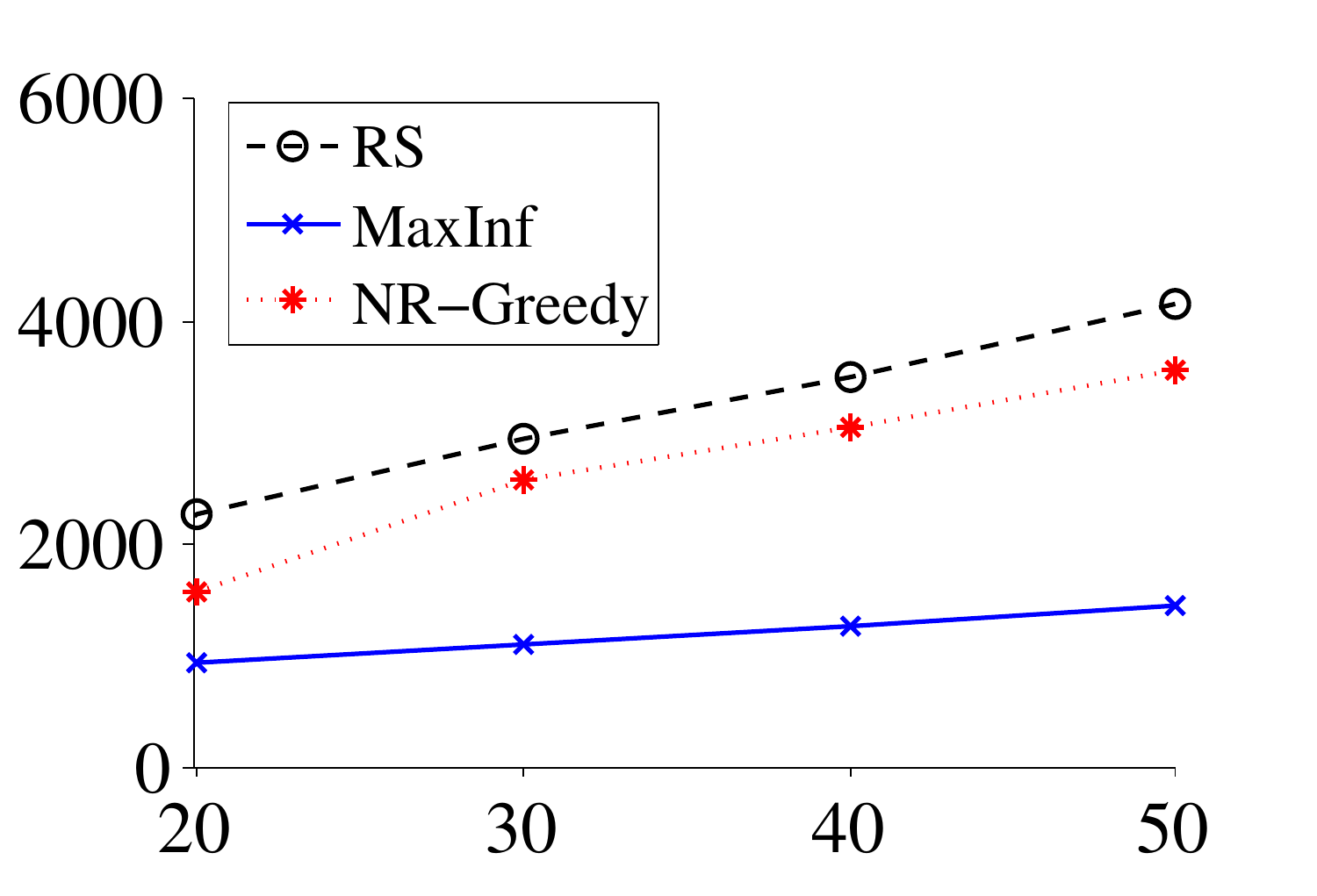}}
\subfloat[{[Youtube, NA]}]{\label{fig: youtube_nei}\includegraphics[width=0.23\textwidth]{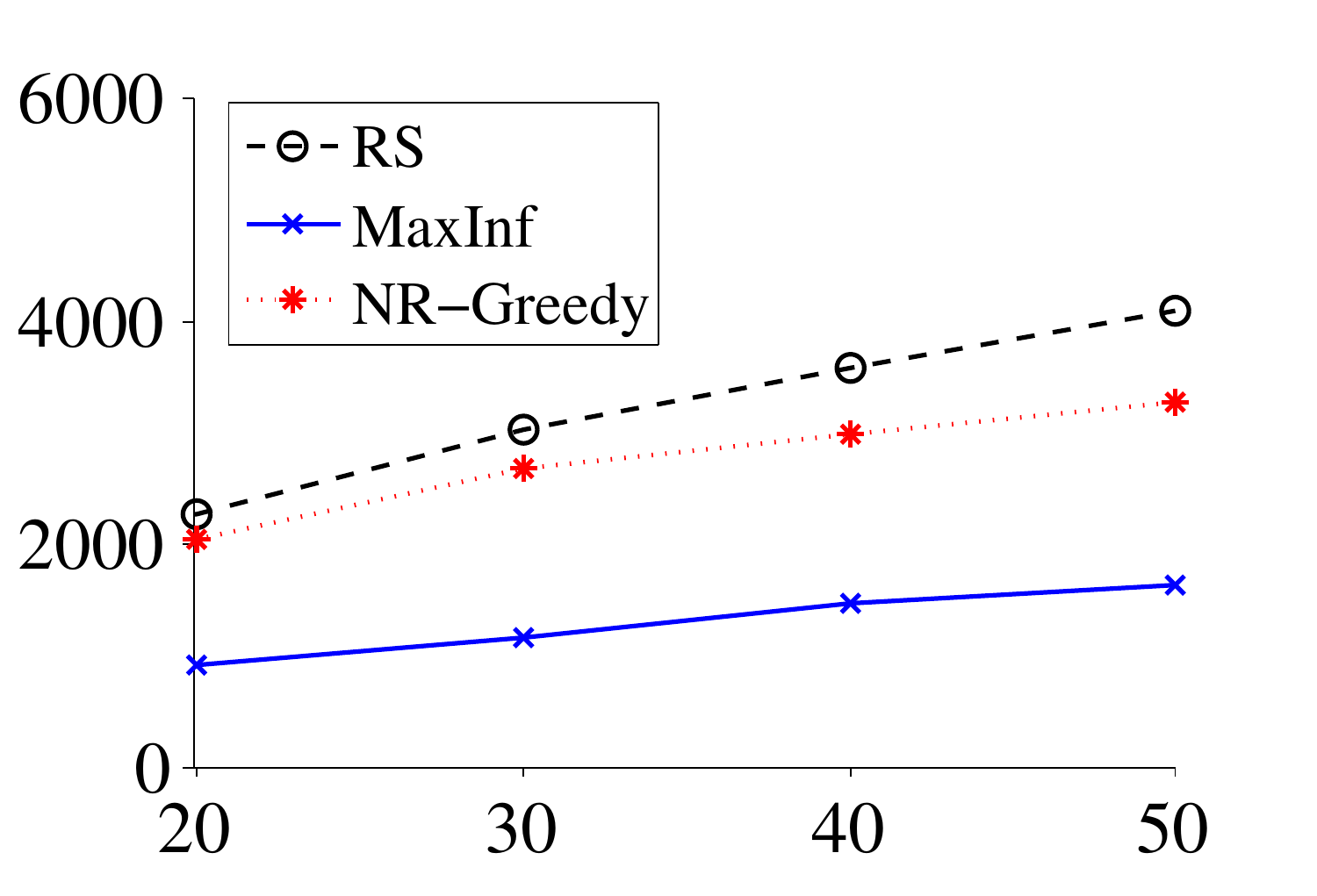}}

\subfloat[{[Youtube, RA]}]{\label{fig: youtube_ran}\includegraphics[width=0.23\textwidth]{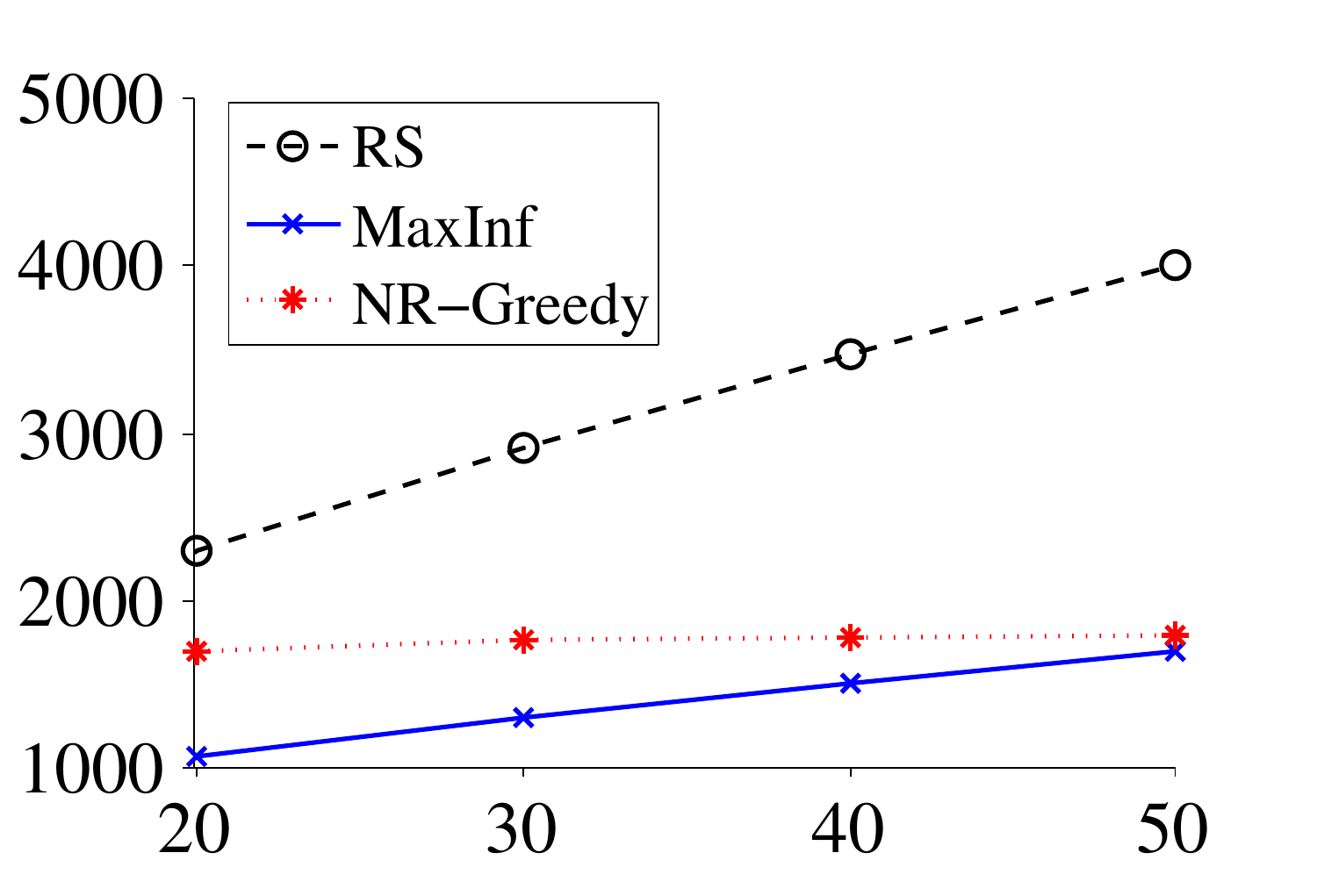}}
\raisebox{0.25\height}{\subfloat{\label{fig: youtube_ratio}\includegraphics[width=0.23\textwidth]{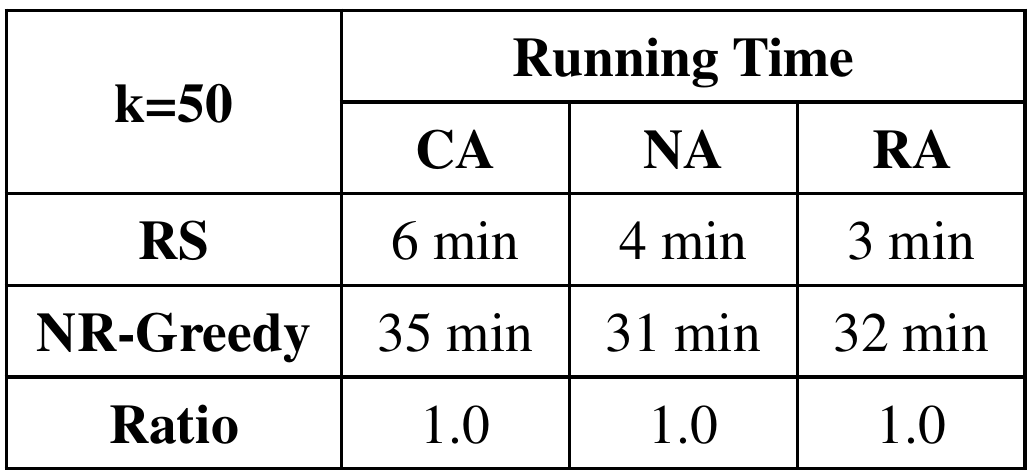}}}
\caption{Main Experimental Results on Youtube.}
\label{fig: main_Youtube}
\end{figure}

\begin{figure}[h!]
\centering
\subfloat[{[Livejournal, CA]}]{\label{fig: live_cas}\includegraphics[width=0.23\textwidth]{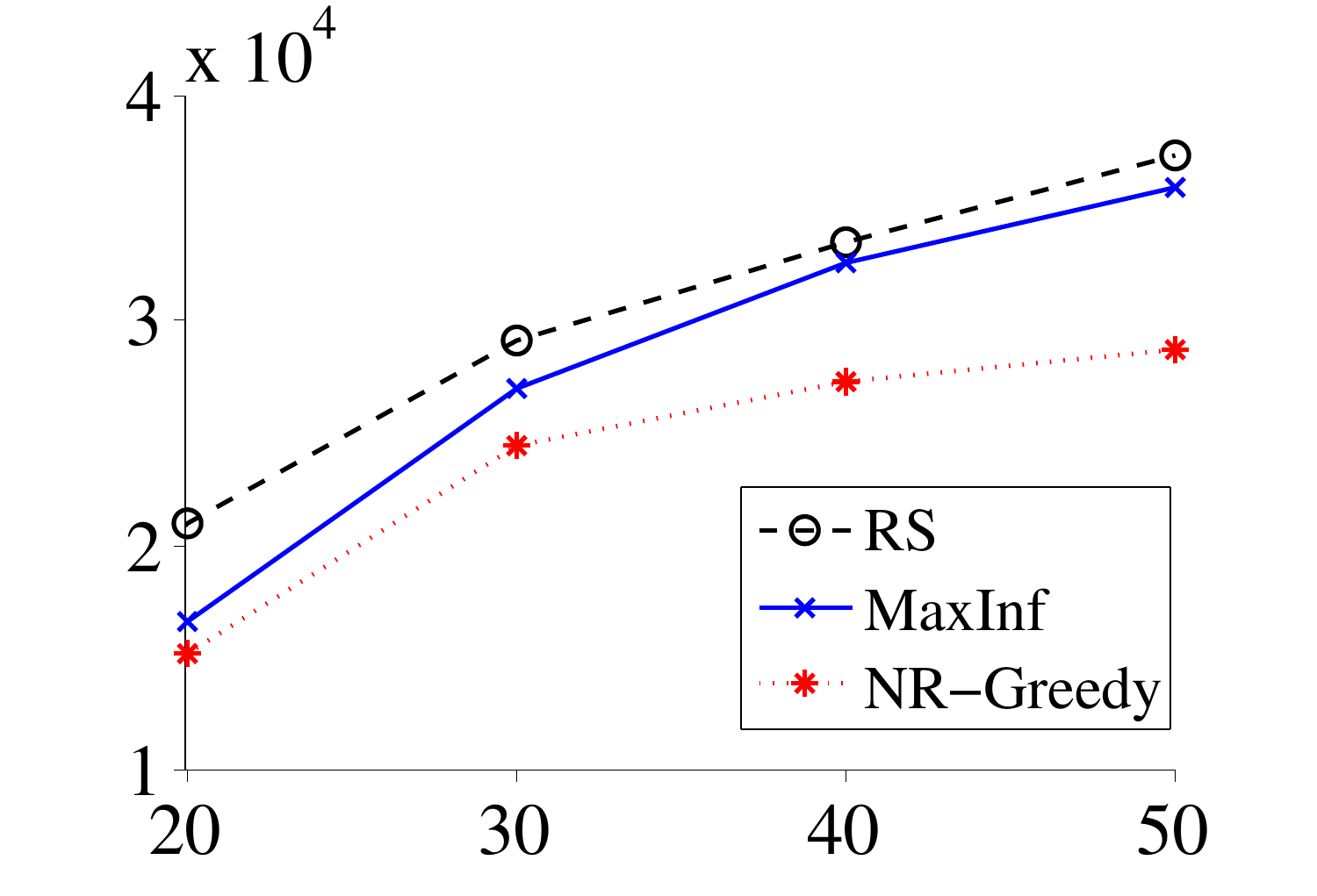}}
\subfloat[{[Livejournal, NA]}]{\label{fig: live_nei}\includegraphics[width=0.23\textwidth]{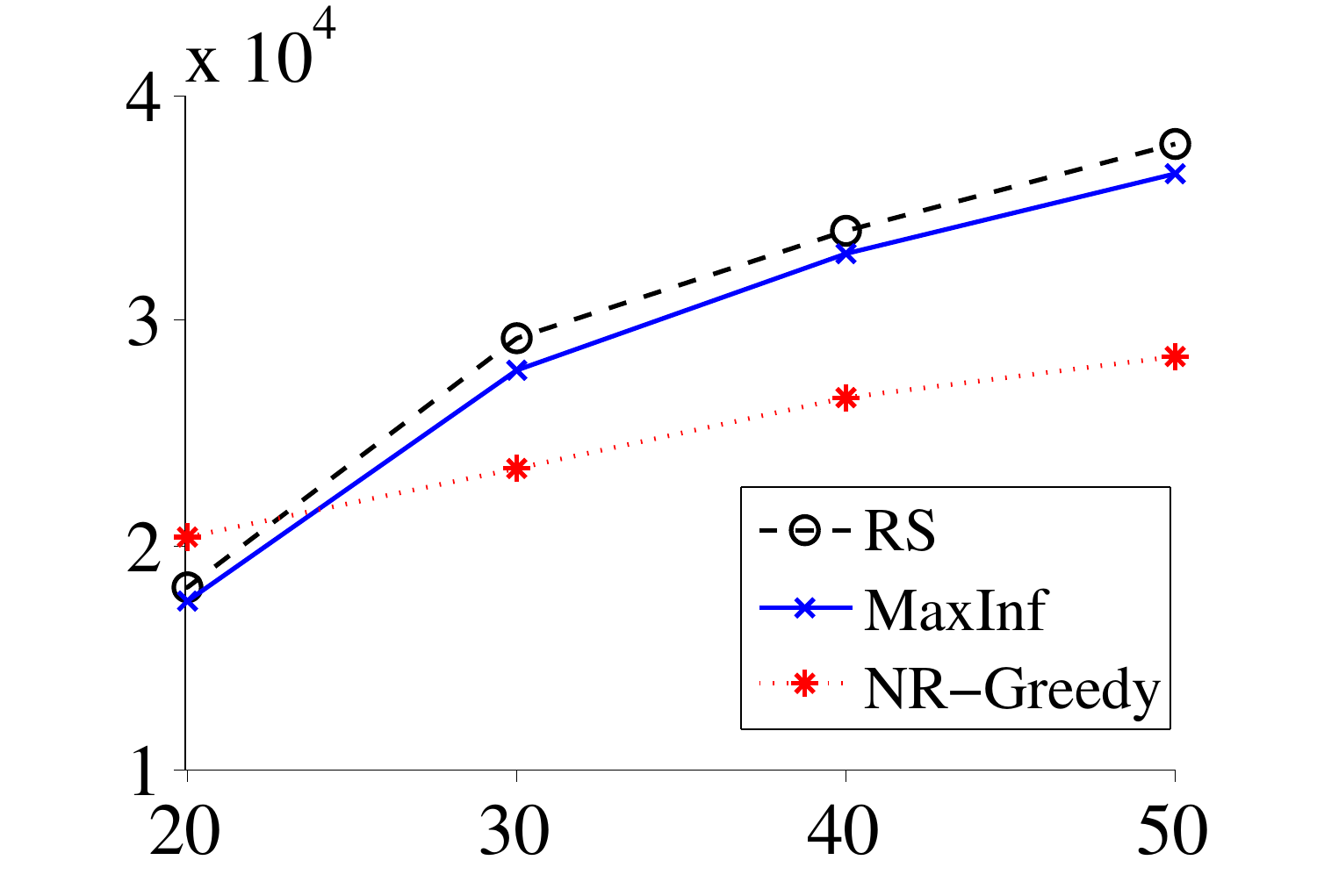}}

\subfloat[{[Livejournal, RA]}]{\label{fig: live_ran}\includegraphics[width=0.23\textwidth]{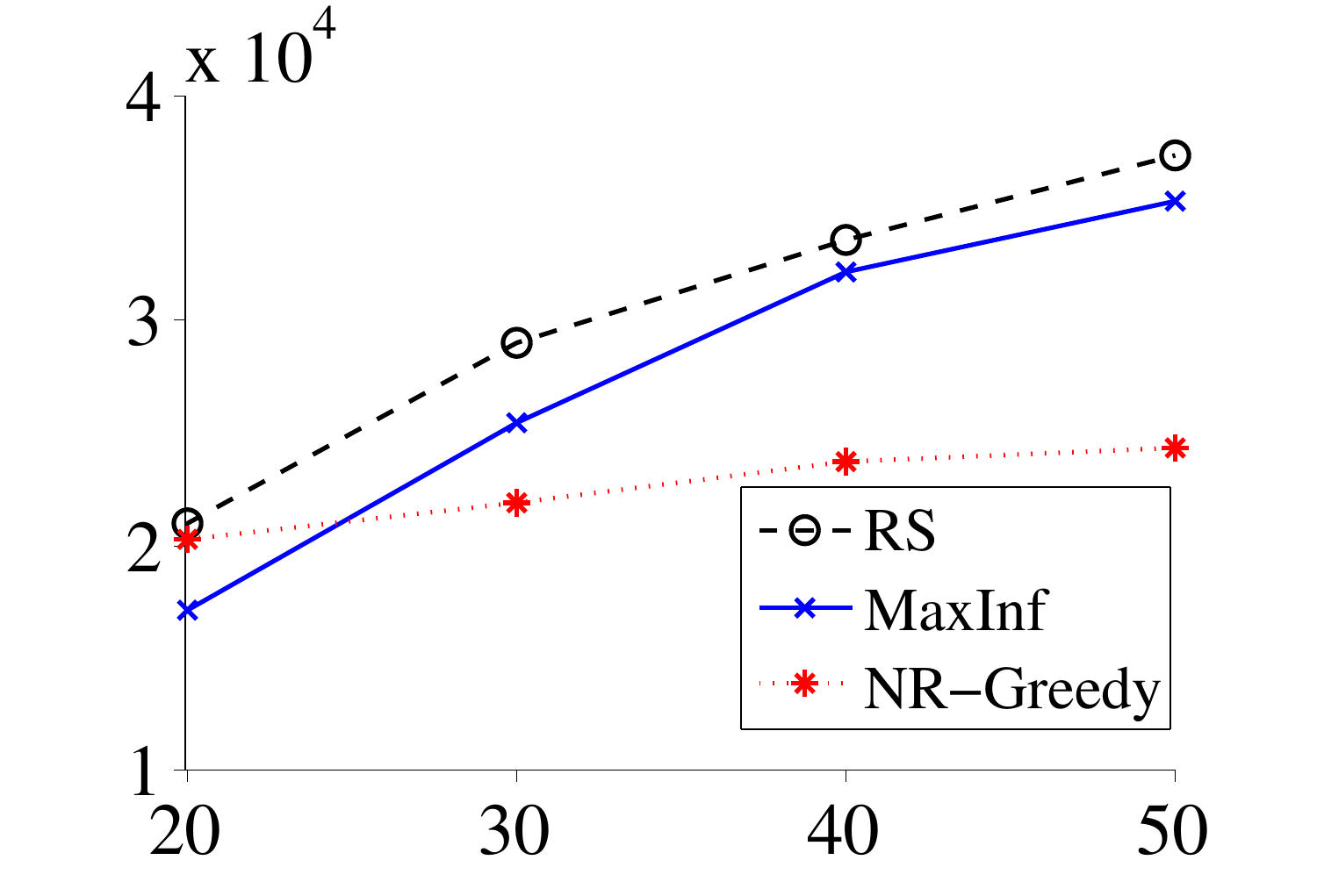}}
\raisebox{0.25\height}{\subfloat{\label{fig: live_ratio}\includegraphics[width=0.23\textwidth]{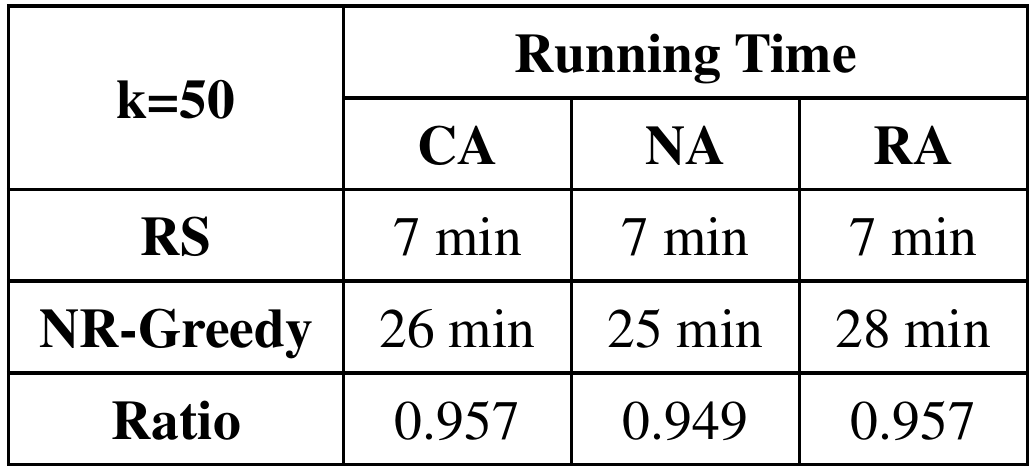}}}

\caption{Main Experimental Results on Livejournal.}
\label{fig: main_Livejournal}
\end{figure}

\begin{table}[h!]
\caption{{Hyper-parameter Study on [Youtube, RA].}}
\centering
\label{table: hyper_youtube}
\begin{tabular}{>{\centering\arraybackslash}p{1.5cm}|>{\centering\arraybackslash}p{1cm}|>{\centering\arraybackslash}p{1cm}|>{\centering\arraybackslash}p{1cm}|>{\centering\arraybackslash}p{1cm}}
\toprule
$N$  &10 &   1000 & 10000  &100000\\
\midrule
Inf   &  4183   & 4142 & 4133 & 4141\\
Time &   5min & 5min & 5.8min & 4.5min\\
\end{tabular}

\vspace{3mm}
\begin{tabular}{>{\centering\arraybackslash}p{1.5cm}|>{\centering\arraybackslash}p{1cm}|>{\centering\arraybackslash}p{1cm}|>{\centering\arraybackslash}p{1cm}|>{\centering\arraybackslash}p{1cm}}
\toprule
$\epsilon$ &   0.1 & 0.2&0.4  &0.9 \\
\midrule
Inf  &  4148   &4131& 4133& 4096\\
Time &  46min &17min    & 2min &43s \\
\end{tabular}

\vspace{3mm}
\begin{tabular}{>{\centering\arraybackslash}p{1.5cm}|>{\centering\arraybackslash}p{1cm}|>{\centering\arraybackslash}p{1cm}|>{\centering\arraybackslash}p{1cm}|>{\centering\arraybackslash}p{1cm}}
\toprule
$K$  &1 &10 &   100 & 1000  \\
\midrule
Inf  &4139 & 4151  &4143  &  4141 \\
Time &13min &   7min    &  6min  & 6min \\
\end{tabular}
\vspace{10mm}
\end{table}

\begin{table}[h!]
\caption{{Hyper-parameter Study on [Livejournal, RA].}}
\centering
\label{table: hyper_livej}
\begin{tabular}{>{\centering\arraybackslash}p{1.5cm}|>{\centering\arraybackslash}p{1cm}|>{\centering\arraybackslash}p{1cm}|>{\centering\arraybackslash}p{1cm}|>{\centering\arraybackslash}p{1cm}}
\toprule
$N$  &10 &   1000 & 10000  &100000\\
\midrule
Inf   &  37311   & 37326 & 37342& 37364\\

Time &   9min & 8min & 8.3min & 9.2min\\
\end{tabular}

\vspace{3mm}
\begin{tabular}{>{\centering\arraybackslash}p{1.5cm}|>{\centering\arraybackslash}p{1cm}|>{\centering\arraybackslash}p{1cm}|>{\centering\arraybackslash}p{1cm}|>{\centering\arraybackslash}p{1cm}}
\toprule
$\epsilon$ &   0.1 & 0.2&0.4  &0.9 \\
\midrule
Inf  &  37338   &37345& 37318& 36016\\

Time &  80min &19min    & 9min &1.8min \\
\end{tabular}

\vspace{3mm}
\begin{tabular}{>{\centering\arraybackslash}p{1.5cm}|>{\centering\arraybackslash}p{1cm}|>{\centering\arraybackslash}p{1cm}|>{\centering\arraybackslash}p{1cm}|>{\centering\arraybackslash}p{1cm}}
\toprule
$K$  &1 &10 &   100 & 1000  \\
\midrule
Inf  &37131 & 37261  &37297  &  37179 \\

Time &32min &   9min    &  9min  & 9min \\
\end{tabular}
\end{table}

\begin{table}[h!]
\centering
\caption{Comparing RS with NR-Greedy.}
\label{table: exp2_hepph}
\begin{tabular}{c|l|l|l|l|l|l}
\toprule
Hepph & \multicolumn{2}{c}{CA} \vline &  \multicolumn{2}{c}{NA} \vline&  \multicolumn{2}{c}{RA} \\
$k=50$ & \makecell{Time } & Inf &  \makecell{Time } & Inf  &\makecell{Time } & Inf \\
\midrule
\multirow{5}{*}{\makecell{NR- \\ Greedy}} & 1 min& 1621 & 1 min& 1518 &1 min& 1209\\
& 9 & 2074 & 9 & 1980 &9& 1582\\
& 19 &2168 & 18 & 2042 & 18 & 1766\\
& 51 &  2252 & 60 & 2155 & 54 & 1922\\
 & 83 & 2267 & 94 & 2210& 180 & 2185\\
\midrule
RS & $1 $ sec & 2261 & $1 $ sec & 2219 &$1 $ sec & 2242
\end{tabular}
\vspace{12mm}
\end{table}

\begin{table*}[h!]
\caption{{Comparing RS with MaxInf.}}
\centering
\label{table: competition}
\begin{tabular}{c|p{1cm}|p{1cm}|p{1cm}|p{1cm}}
\toprule
 \multirow{1}{*}{}&  k=20  & k=30 &  k=40 & k=50    \\
\midrule
1\%& 0.935& 0.957 & 0.941 &0.919  \\
5\%& 0.981 & 0.839& 0.823 &0.797  \\
10\%&0.996& 0.799 &0.777 &0.750
\end{tabular}
\end{table*}
\begin{figure*}[h!]
\centering
\subfloat[{[Livejournal, CA, $1\%$]}]{\label{fig: live1_2}\includegraphics[width=0.40\textwidth]{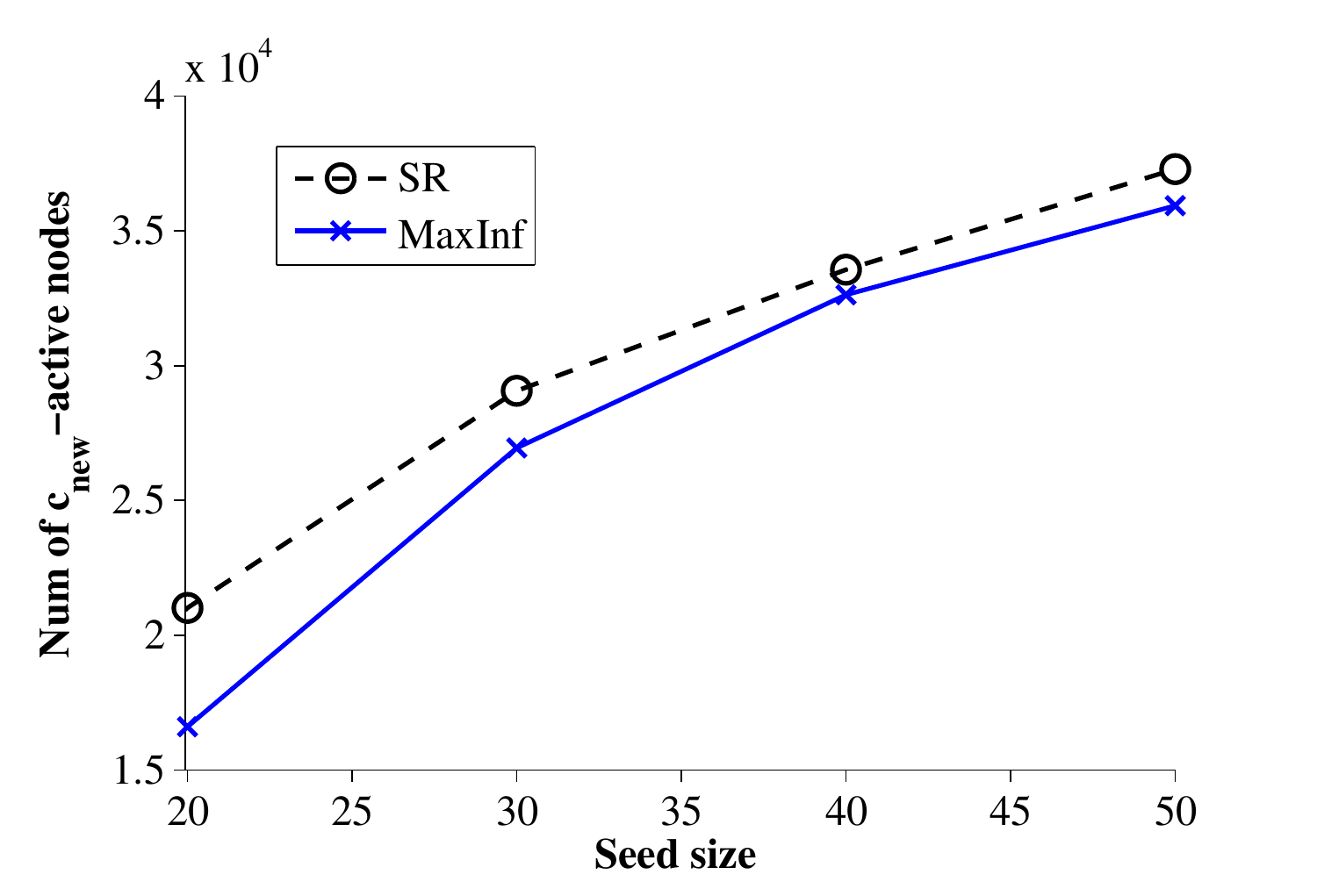}}
\subfloat[{[Livejournal, CA, $5\%$]}]{\label{fig: live5_2}\includegraphics[width=0.40\textwidth]{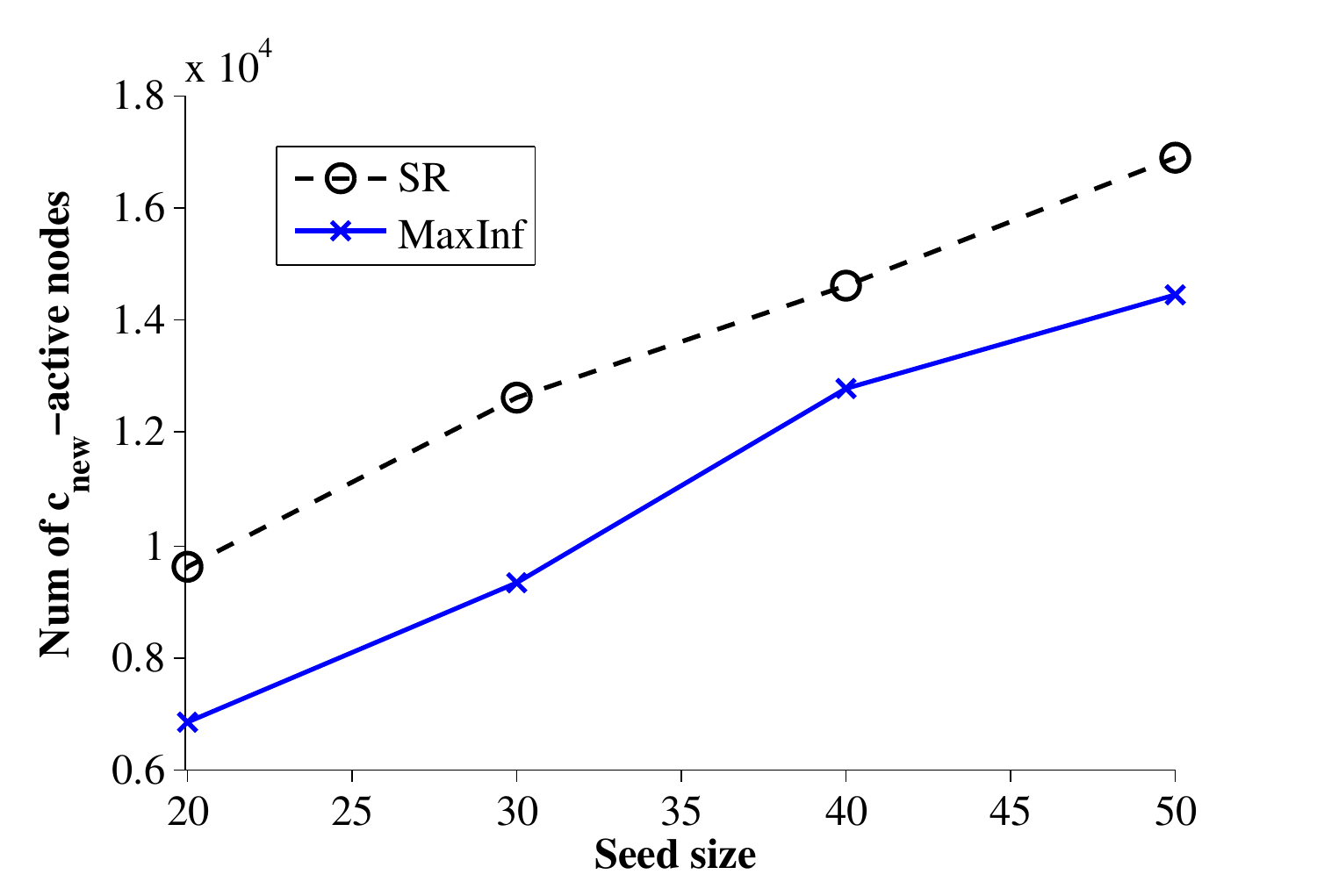}}
\caption{Comparing RS with MaxInf on Livejournal.}
\vspace{-5mm}
\label{fig: rswithmaxinf_live}
\end{figure*}
\begin{figure*}[h!]
\centering
\captionsetup[subfigure]{labelformat=empty}
\subfloat{\label{fig: overlap_hepph_ca_l}\includegraphics[width=0.30\textwidth]{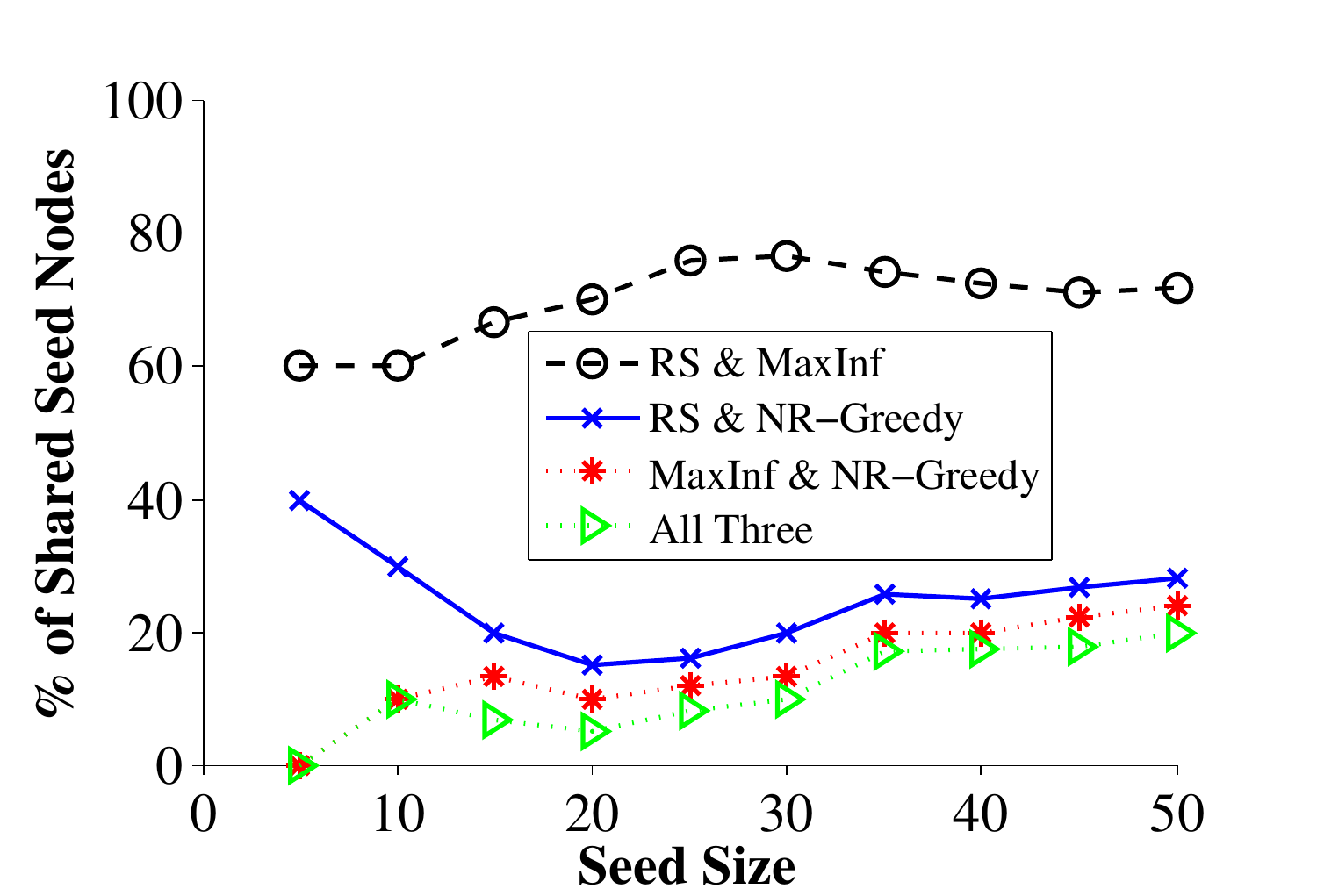}}
\subfloat{\label{fig: overlap_hepph_na_l}\includegraphics[width=0.30\textwidth]{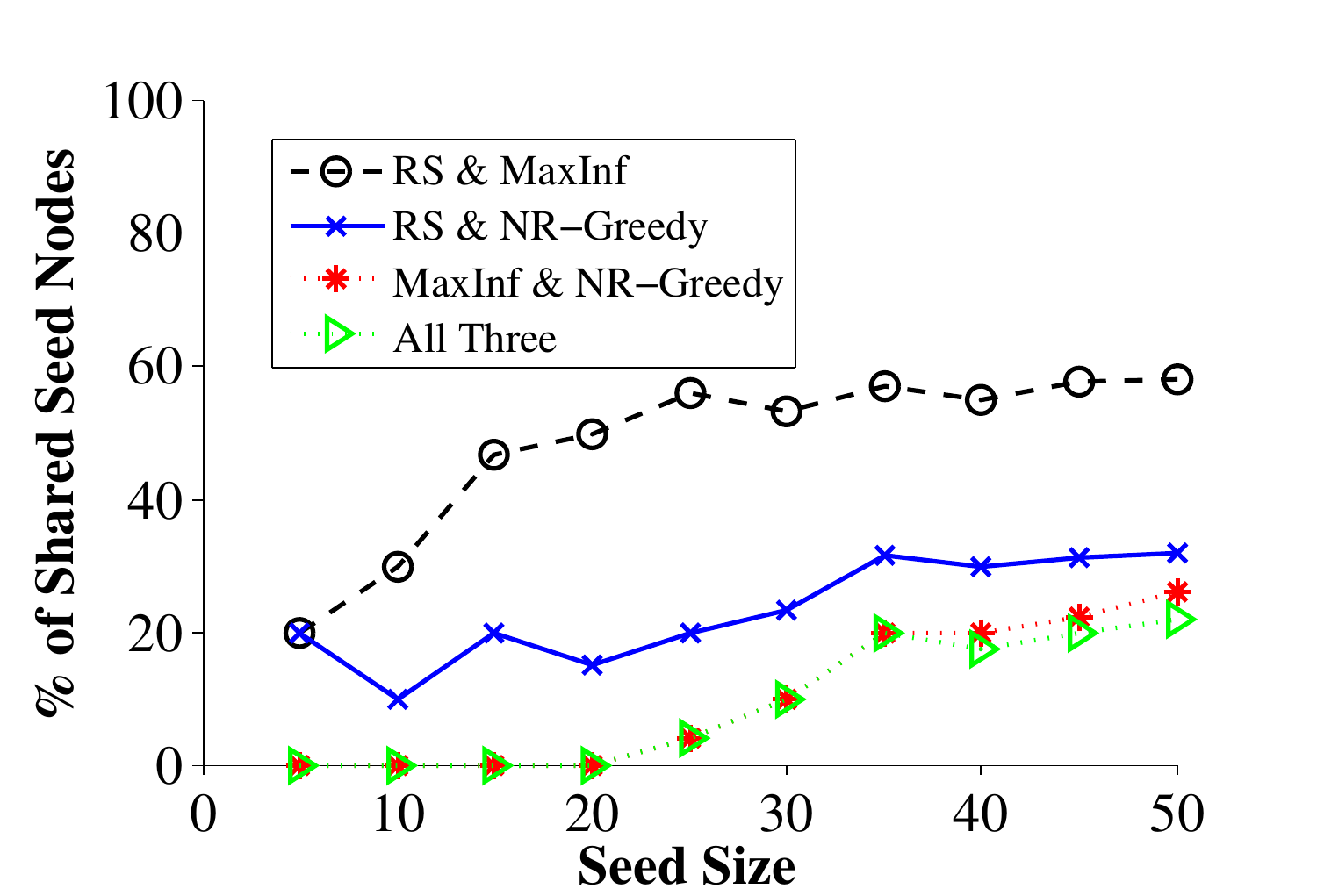}}
\subfloat{\label{fig: overlap_hepph_ra_l}\includegraphics[width=0.30\textwidth]{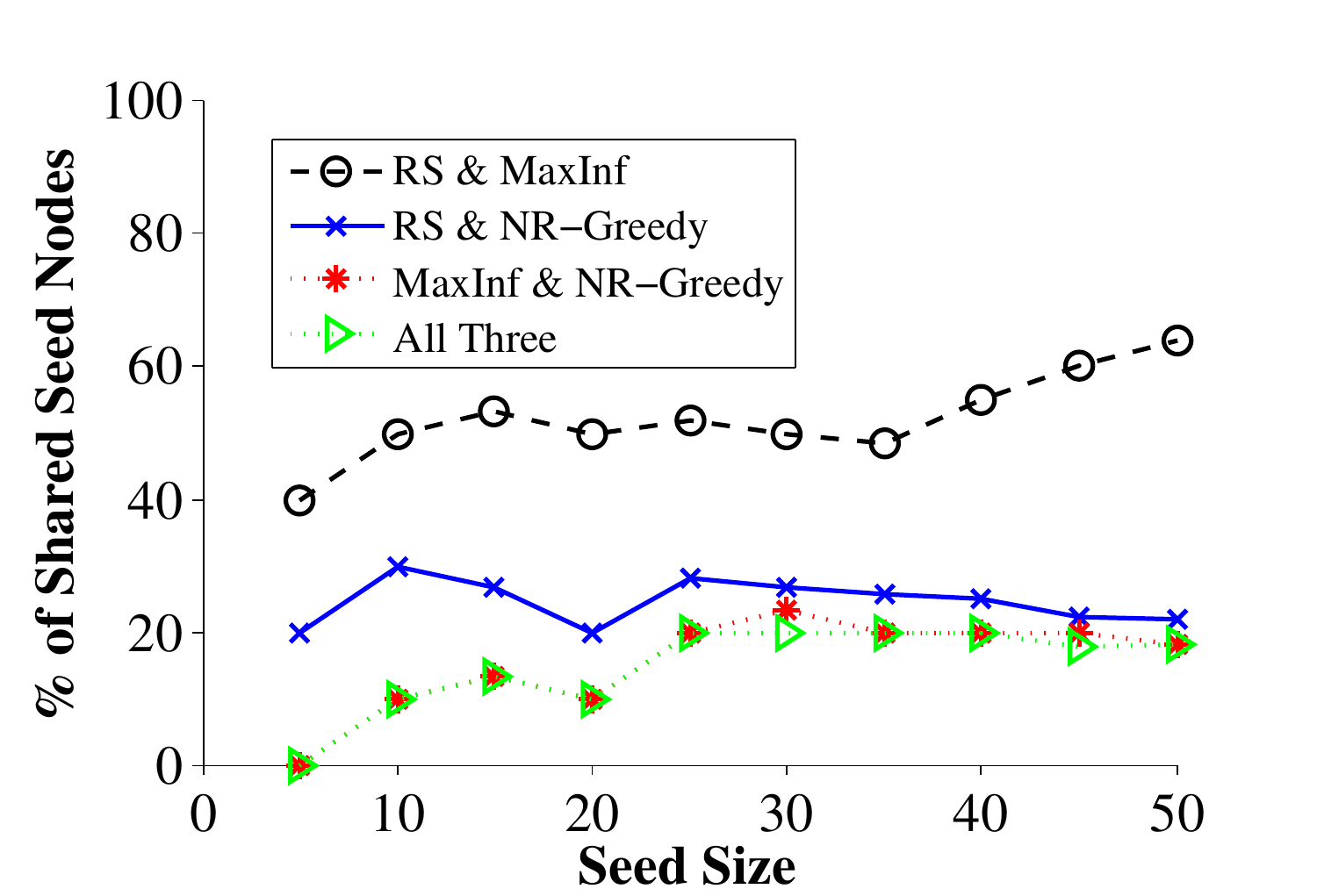}}

\raisebox{0.25\height}{\subfloat[ {[Hepph, CA, k=50]}]{\label{fig: overlap_hepph_ca_r}\includegraphics[width=0.15\textwidth]{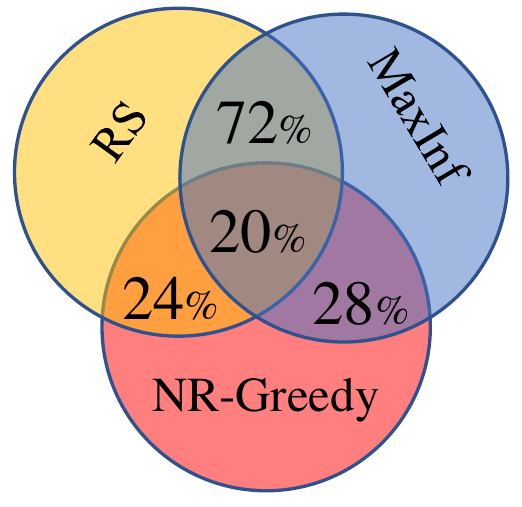}}\hspace{+15mm}}
\raisebox{0.25\height}{\subfloat[ {[Hepph, NA, k=50]}]{\label{fig: overlap_hepph_na_r}\includegraphics[width=0.15\textwidth]{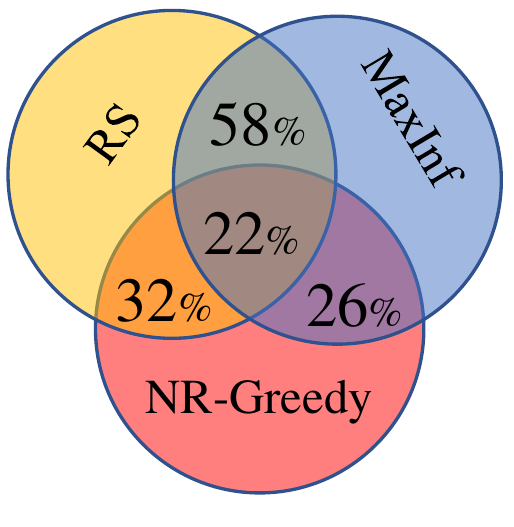}}\hspace{+15mm}}
\raisebox{0.25\height}{\subfloat[ {[Hepph, RA, k=50]}]{\label{fig: overlap_hepph_ra_r}\includegraphics[width=0.15\textwidth]{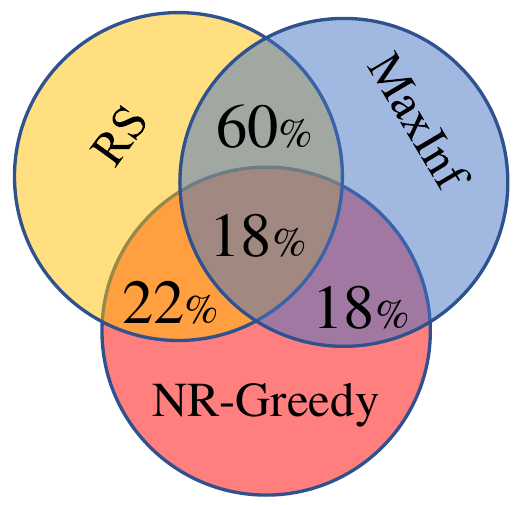}}}
\caption{Seed Nodes Overlap on Hepph.}
\label{fig: overlap}
\end{figure*}

\begin{table*}[h!]
\caption{Results under RA with Large Settings.}
\label{table: extreme}
\centering
\begin{tabular}{>{\centering\arraybackslash}p{1.5cm}|>{\centering\arraybackslash}p{1.2cm}|>{\centering\arraybackslash}p{0.7cm}|>{\centering\arraybackslash}p{0.7cm}|>{\centering\arraybackslash}p{0.7cm}|>{\centering\arraybackslash}p{0.7cm}}
\multicolumn{2}{c|}{\textbf{Youtube}} & 50 & 150 &  500   & 1000  \\
\midrule
\multirow{2}{*}{Inf ($10^3$)} & MaxInf  &1.6 & 3.2 & n/a& n/a \\
&RS& 4.0 & 6.5 & 7.8 & 11.7 \\
\midrule
\multirow{2}{*}{Mem (G)} & MaxInf  &41.3 & 111 & n/a & n/a\\
&RS & 4.5 & 6.1 & 14.6 & 22.0\\
\midrule
\multirow{2}{*}{Time (min)} & MaxInf  &31 & 56& n/a & n/a\\
&RS & 4 & 9 & 31 & 45
\end{tabular}\hspace{+5mm}
\begin{tabular}{>{\centering\arraybackslash}p{1.5cm}|>{\centering\arraybackslash}p{1.2cm}|>{\centering\arraybackslash}p{0.7cm}|>{\centering\arraybackslash}p{0.7cm}|>{\centering\arraybackslash}p{0.7cm}|>{\centering\arraybackslash}p{0.7cm}}
\multicolumn{2}{c|}{\textbf{Livejournal}} & 200 & 400 &  800   & 1000  \\
\midrule
\multirow{2}{*}{Inf ($10^3$)} & MaxInf  &49.5& 91.3 &157.8 & 188.2 \\
&RS& 90.9 & 170.2 & 289.7 & 334.9 \\
\midrule
\multirow{2}{*}{Mem (G)} & MaxInf  &9.5 & 13.1 & 15.5& 24.9\\
&RS& 7.4 & 10.6 & 15.0 & 23.1\\
\midrule
\multirow{2}{*}{Time (min)} & MaxInf  &13 & 19& 25 & 27\\
&RS & 15 & 23 & 31 & 33
\end{tabular}

\end{table*}

\clearpage
\section{Proofs}
\label{sec: proof}

\subsection{Proof of Lemma \ref{lemma: max_mim}}
\label{proof: lemma: max_mim}

A reduction is given as follows.

\textbf{Reduction Construction.} For an instance $(\overline{G}=(\overline{V},\overline{E}), \overline{k} )$ of DkS, we construct an instance of Max-MIM as illustrated in Fig. \ref{fig: reduction_max}. For each $v \in \overline{V}$, we add one node $x_v$ to the graph. For each $(u, v) \in \overline{E}$, we add one node $y_{(u, v)}$ and two edges $(x_v, y_{(u, v)})$ and $(x_u, y_{(u, v)})$. Finally, we add a node $z$ with an edge $(z, y_{(u, v)})$ for each node $y_{(u, v)}$. We assume that there are two existing cascades $c_1$ and $c_2$, where the seed set of $c_1$ is $\{x_v| v\in \overline{V}\}$ and the seed set of $c_2$ is $\{z\}$. We consider the cascade-based activation function defined in Def. \ref{def: cascade_function}. For each $x \in \{x_v| v\in \overline{V}\}$, we assume that $F_x(\{c_1,c_{new}\})=c_{new}$. For each $y \in \{y_{(u,v)}| (u,v) \in \overline{E}\}$, we assume that $F_y(\{c_2,c_{new}\})=c_{new}$ and $F_y(\{c_2, c_1, c_{new}\})=c_2$. We set $p_e=1$ for each edge, $V_{c}=\{x_v| v\in \overline{V}\}$, and $k=\overline{k}$.

\begin{figure}[t]
	\begin{center}
		\includegraphics[width=0.46\textwidth]{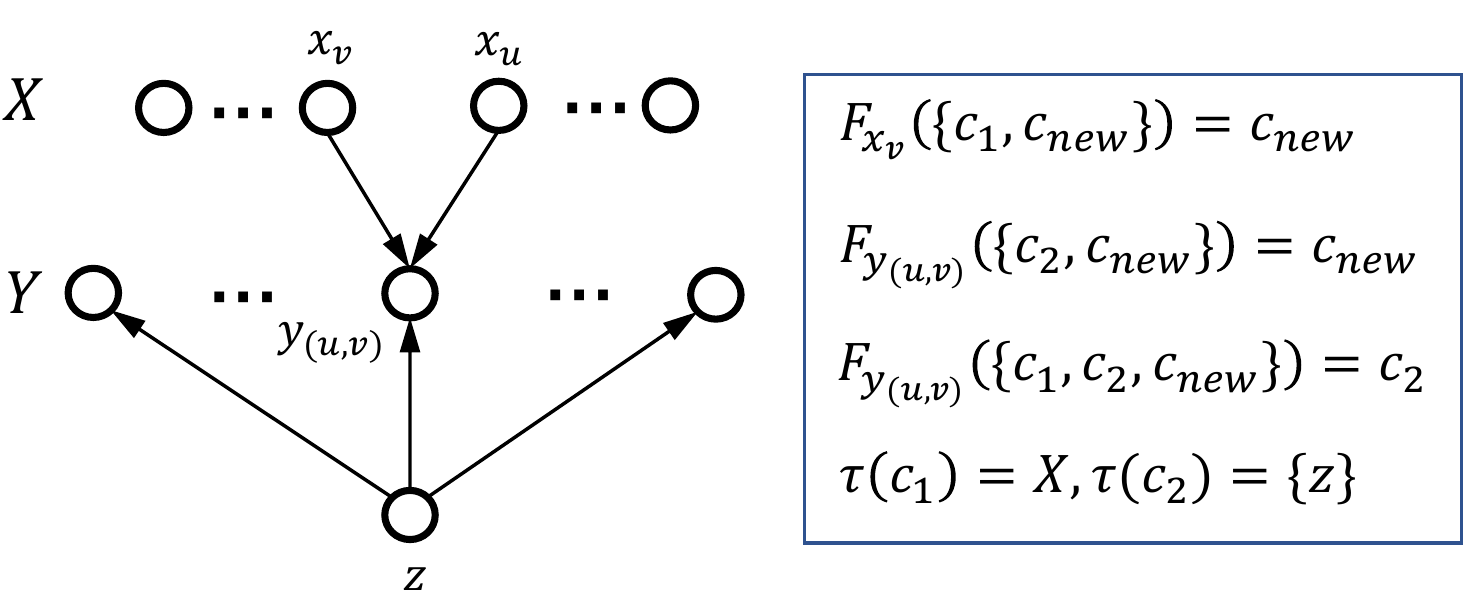} 
	\end{center} 
	\caption{Reduction from DkS to Max-MIM.}
	\label{fig: reduction_max}
\end{figure}
\begin{proof}

\textbf{Analysis.} Each seed set $S \subseteq V_c$ of $c_{new}$ with size $k$ corresponds to a $\overline{k}$-subgraph of DkS. Let $f^*(S)$ be the number of the edges in the subgraph. For each node $x$ in $\{x_v| v\in \overline{V}\}$, they will be $c_{new}$-active if and only if they are selected as a seed node. For each node $y_{(u,v)}$ in $\{y_{(u,v)}| (u,v) \in \overline{E}\}$, they will be $c_{new}$-active if and only if both $x_u$ and $x_v$ are selected as the seed nodes due to the activation function at $y_{(u,v)}$. Thus, we have $f(S)=k+f^*(S)$ and both instances have the same optimal solution denoted as $S_{opt}^{*}$. With loss of generality, we only consider the candidate $S$ with $f^*(S) \geq k$.\footnote{This is because computing a $k$-subgraph with at least $k$ edges, if any, can be done polynomial time \cite{asahiro2002complexity}.} Now suppose $S^*$ is an $\alpha$-approximation to Max-MIM for some $\alpha>1$, we have 
\begin{align*}
&\alpha \cdot  f(S^*)\geq f(S^*_{opt}) \\
\Leftrightarrow & \alpha \cdot \big(k+f^*(S^*)\big)  \geq \big(k+f^*(S_{opt}^*)\big)\\
\Rightarrow & 2  \alpha \cdot f^*(S^*)  \geq   f^*(S_{opt}^*).
\end{align*}

Thus proved.
\end{proof}

\subsection{Proof of Lemma \ref{lemma: reduction}}
\label{proof: lemma: reduction}
\begin{proof}
We prove this by a reduction. 

\begin{figure}[t]
\begin{center}
\includegraphics[width=0.44\textwidth]{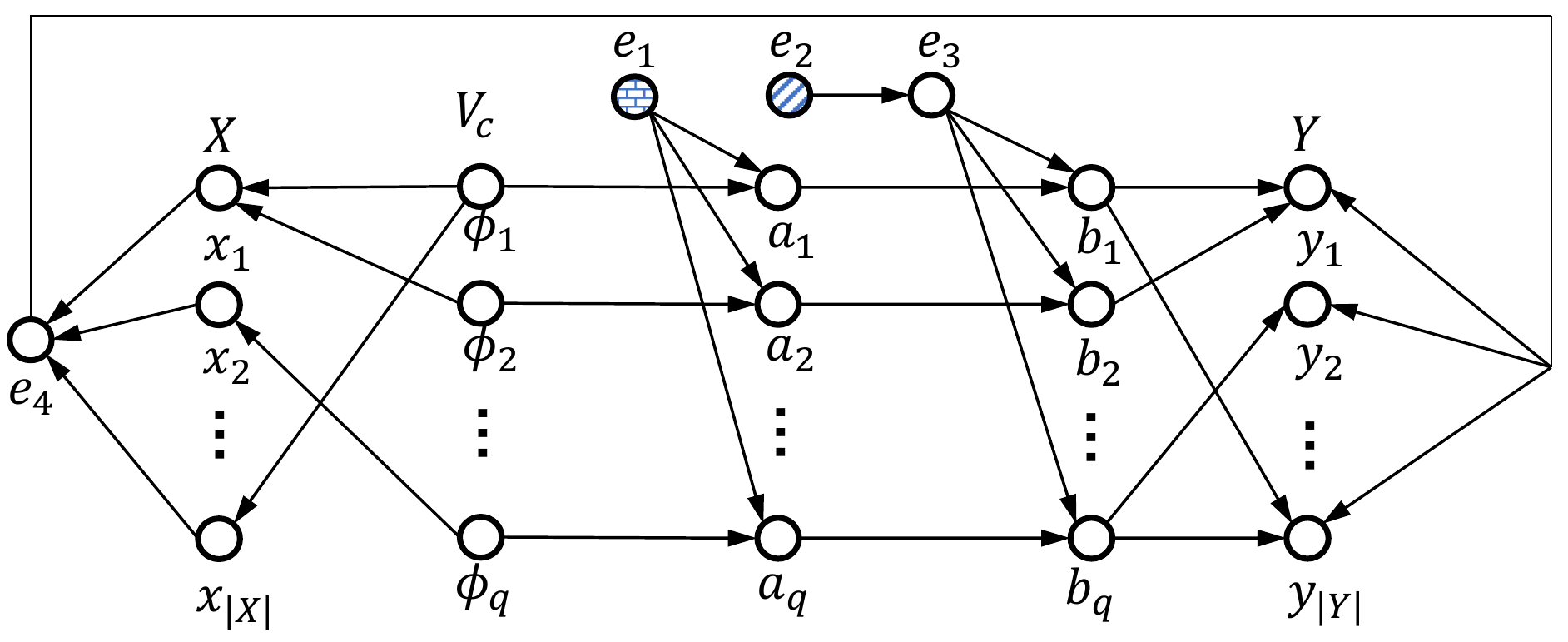} 
\end{center} 
\vspace{-3mm}
\caption{Reduction from k$\pm$PSC to Min-MIM.}
\label{fig: reduction}
\vspace{-4mm}
\end{figure}

\textbf{Reduction Construction.} Consider an arbitrary instance $(X,Y,\Phi)$ of k$\pm$PSC with $X=\{{x}_1,...,{x}_{|X|}\}$, $Y=\{{y}_1,...,{y}_{|Y|}\}$, $\Phi=\{{\phi}_1,...,{\phi}_q\}$ and $\overline{k}$. In what follows, we construct an instance of Min-MIM as shown in Fig. \ref{fig: reduction}. For each element ${x}_i$ in $X$, we add one node $x_i$ to the graph.\footnote{Without causing any confusion, we use the same symbol for the items in two instances that refer to the same object} For each element ${\phi}_i$ in $\Phi$, we add three nodes $\phi_i$, $a_i$ and $b_i$ to the graph. For each element ${y}_i$ in $Y$, we add one node $y_i$ to the graph. Finally, we add four nodes $e_1, e_2, e_3$ and $e_4$. We add two edges $(\phi_i,a_i)$ and $(a_i,b_i)$ for each triple ($\phi_i$, $a_i$, $b_i$). For each pair of $i$ and $j$, we add an edge from $\phi_i$ to $x_j$ if and only if ${x}_j \in {\phi}_i$, and, add an edge from $b_i$ to $y_j$ if and only if ${y}_j \in {\phi}_i$. For each $i$, we add two edges $(e_1,a_i)$ and $(e_3,b_i)$. For each $x_i$, we add one edge $(x_i,e_4)$. For each $y_i$, we add one edge $(e_4,y_i)$. Finally, we add one edge $(e_2,e_3)$. We assume that there are three cascades $c_1$, $c_2$ and $c_{new}$. The seed sets of $c_1$ and $c_2$ are $\tau(c_1)=\{e_1\}$ and $\tau(c_2)=\{e_2\}$, respectively. We consider the total-ordered-based activation function defined in Def. \ref{def: cascade_function}, and define the settings: $F_{a_i}(c_{new})>F_{a_i}(c_1)$ for each $a_i$; $F_{b_i}(c_1)> F_{b_i}(c_2)>F_{b_i}(c_{new})$ for each $b_i$; $F_{y_i}(c_2)>F_{y_i}(c_{new})> F_{y_i}(c_1)$ for each $y_i$. We set $p_e=1$ for each edge $e\in E$, $V_{c}=\{\phi_1,...,\phi_q\}$, and $k=\overline{k}$.

\textbf{Analysis.} According to the construction, both instances have the same candidate space. For a subset ${\Phi}^* \subseteq \Phi$, we use $\Phi^*$ to denote the corresponding set of nodes. Supposing, a particular seed set $\Phi^*$ of $c_{new}$ is selected, let us examine $\overline{f}(\Phi^*)$ by analyzing which nodes will not be activated by $c_{new}$. We consider the nodes group by group.

\textit{Group 1.} The nodes in $\{\phi_1,...,\phi_q\}$ will be activated by $c_{new}$ if and only if they are selected as seed nodes. Therefore, the contribution to $\overline{f}(\Phi^*)$ from this part is $q-k$.

\textit{Group 2.} Each node $x_i$ in $\{x_1,...,x_{|X|}\}$ can be activated by $c_{new}$ if and only if there is an edge from a seed node in $\Phi^*$ to $x_i$, i.e, ${x}_i \in {\phi}_j$ for some ${\phi}_j \in {\Phi}^*$. Therefore, the contribution to $\overline{f}(\Phi^*)$ from $\{x_1,...,x_{|X|}\}$ is $|X\setminus \big(\cup_{\phi \in {\Phi}^*} \phi\big)|$.

\textit{Group 3.} Each node $a_i$ in $\{a_1,...,a_{q}\}$ can be activated by $c_{new}$ if and only if $\phi_i$ is in $\Phi^*$ since otherwise they will be activated by $c_1$ from $e_1$. Therefore, the contribution from this part is $q-k$.

\textit{Group 4.} Each node $b_i$ in $\{b_1,...,b_{q}\}$ can never be activated by $c_{new}$. Note that $a_i$ will activated by either $c_1$ or $c_{new}$. There are two sub-cases:

Case 1. If $a_i$ is activated by $c_{new}$, $b_i$ will  be activated by $c_{new}$ and $c_2$ at the same time, and consequently $b_i$ will be $c_2$-active due to the activation function at $b_i$. This case happens if and only if $\phi_i \in \Phi^*$.

Case 2. If $a_i$ is activated by $c_{1}$, $b_i$ will  be activated by $c_{1}$ and $c_2$ at the same time, and consequently $b_i$ will be $c_1$-active due to again the activation function at $b_i$. This case happens if and only if $\phi_i \notin \Phi^*$.

Thus, the contribution to $\overline{f}(\Phi^*)$ from this part is $q$.

\textit{Group 5.} Each node $y_i$ in $\{y_1,...,y_{|Y|}\}$ will be either $c_{new}$-active or $c_2$-active, depending on whether or not ${y}_i$ is in ${\phi}_j$ for some ${\phi}_j \in {\Phi}^*$. There are two sub-cases:

Case 1: If ${y}_i$ is in ${\phi}_j$ for some ${\phi}_j \in {\Phi}^*$, $b_j$ will be $c_2$-active according to the analysis of $b_j$. The activation function at $y_i$ will make $y_i$ be $c_2$-active.

Case 2: If ${y}_i$ is not in any ${\phi}_j$ in ${\Phi}^*$, all the nodes in $\{b_j\}$ connected to ${y}_i$ will be $c_1$-active and in this case $y_i$ will be simultaneously activated by $c_1$ from $\{b_j\}$ and $c_{new}$ from $e_4$. Therefore, it will be $c_{new}$-active due to again the activation function at $y_i$. 

As a result, the total contribution to $\overline{f}(\Phi^*)$ from this part is $|Y\cap \big(\cup_{\phi \in {\Phi}^*} \phi\big)|$.

\textit{Group 6.} Finally, $e_4$ will always be $c_{new}$-active, and, $e_1$, $e_2$, and $e_3$ can never be $c_{new}$-active.

Therefore, for each candidate $\Phi^* \subseteq \Phi$ of MIM, we have its objective value ${f}(\Phi^*)$ as 
\begin{align*}
&\overline{f}(\Phi^*)\\
=&3q-2k+3+ |X\setminus \big(\cup_{\phi \in {\Phi}^*} \phi\big)|+|Y\cap \big(\cup_{\phi \in {\Phi}^*} \phi\big)| \\
 =&\beta+cost({\Phi}^*),
\end{align*}
where we have defined $\beta \define 3q-2k+3$ which is independent of $\Phi^*$. It is clear that both instances have the same optimal solution denoted as ${\Phi}^*_{opt}$. Now suppose ${\Phi}^*$ is an approximation solution to MIM for a certain $\alpha>1$. We have

\begin{align*}
&&\overline{f}(\Phi^*)&\leq& &\alpha \cdot \overline{f}(\Phi^*_{opt}).\\
&\Leftrightarrow& \beta+cost({\Phi}^*)  &\leq&& \alpha \cdot (\beta+cost({\Phi}^*_{opt})) \\
&\Leftrightarrow& cost({\Phi}^*)  &\leq&& \big( \frac{(\alpha-1) \cdot \beta }{cost({\Phi}^*_{opt})} +\alpha \big) \cdot cost({\Phi}^*_{opt}) \\
&\Rightarrow& cost({\Phi}^*)  &\leq&& \big( \alpha\cdot \beta +\alpha \big) \cdot  cost({\Phi}^*_{opt}) 
\end{align*}
Note that $\alpha\cdot \beta +\alpha=\Theta(\alpha\cdot q)$ and $|V_c|=q$. Thus, proved.
\end{proof}

\subsection{Proof of Theorem \ref{theorem: key_1}}
\label{proof: theorem: key_1}
Theorem \ref{theorem: key_1} is established on the following two claims.

\begin{claim}
$\E[\underline{g}(S , \R_v)] \leq \E[g(S ,\R_v)]\leq E[\overline{g}(S ,\R_v)]$
\end{claim}
\begin{proof}
It suffice to show that $\underline{g}(S , R_v) \leq g(S ,R_v)\leq \overline{g}(S ,R_v)$ for each $R_v=(G_v, \underline{S}_v, \overline{S}_v)$. Since these functions are binary valued, it suffices to prove $\underline{g}(S , R_v)=1 \Rightarrow g(S ,R_v)=1$ and $g(S ,R_v)=1 \Rightarrow \overline{g}(S ,R_v)=1$.

Suppose $\underline{g}(S , R_v)=1$. We have $S \cap \underline{S}_v \neq \emptyset$ and therefore, $v$ will be $c_{new}$-active in $\M_{R_v}$ under $S$ because $\underline{S}_v$ are the nodes with distance to $v$ shorter than that from any node in $\cup_{c \in C_e} \tau(c)$, implying that $g(S ,R_v)=1$.

Suppose $g(S ,R_v)=1$. Then at least one node in $V(G_v)$ (i.e. $\overline{S}_v$) is selected as the seed node of $c_{new}$, which means $\overline{g}(S ,R_v)=1$
\end{proof}

\begin{claim}
$\E[g(S ,\R_v)]=f_v(S)$.
\end{claim}
\begin{proof}
Let us first introduce some notations. Following \cite{kempe2003maximizing}, we speak of an edge $(u,v)$ as begin \textit{live} (resp., \textit{dead}) if $u$ can (resp., cannot) activate $v$. Therefore, the basic event space is given by the pairs $\T=\{\big(E_t(T),E_f(T)\big): E_t(T), E_f(T) \subseteq E, E_t(T) \cap E_f(T) =\emptyset\}$ where $E_t(T)$ and $E_f(T)$ intuitively denote the sets of the live and dead edges, respectively. We call each $T \in \T$ a \textit{realization}, and use $\Pr[T]$ to denote the probability that it can be realized where $\Pr[T]=\prod_{e\in E_t(T)} p_e \prod_{e\in E_f(T)} (1-p_e)$.  Abusing the notation slightly, for each $T \in T$, let us define
\[g(S,T) \define
  \begin{cases}
  1 &  \hspace{0mm} \hspace{-0.5mm} \text{$v$ is $c_{new}$-active under $S$ in $T$} \\
  0 & \hspace{0mm} \hspace{-0.5mm} \text{otherwise } 
  \end{cases},\]
and therefore we have $f_v(S)=\sum_{T \in \T} \Pr[T] \cdot g(S,T)$. On the other hand, let $\chi(v)$ be all the possible $R_v$ return by Alg. \ref{alg: rrtuple_v}, and we use $\Pr[R_v]$ to denote probability that $R_v$ can be generated. For each $R_v=(G_v, \underline{S}_v, \overline{S}_v) \in \chi(v)$, we use $E_t(R_v)$ (resp., $E_f(R_v)$) to denote the set of the edges in $E$ that pass (resp., do not pass) the test in Line 16 in Alg. \ref{alg: rrtuple_v}. Therefore, $\Pr[R_v]$ can be explicitly given as $\Pr[R_v]=\prod_{e\in E_t(R_v)} p_e \prod_{e\in E_f(R_v)} (1-p_e)$.

For a pair $T \in \T$ and $R_v=(G_v, \underline{S}_v, \overline{S}_v) \in \chi(v)$, we say $T$ is compatible with $R_v$ if and only if $E_t(T) \subseteq E_t(R_v)$ and $E_f(T) \subseteq E_f(R_v)$, denoted as $T \sim R_v$. One can easily check that each $T \in \T$ is compatible to exactly one $R_v$ in $\chi(x)$, and thus we can partition $\T$ as $\big\{ \{T \in \T: T \sim R_v \} \big\}_{R_v \in \chi(x)}$. As a result, we have $f_v(S)=\sum_{R_v \in \chi(v)} \sum_{T \sim R_v} \Pr[T] \cdot g(S,T)$. For each pair of $T$ and $R_v$ with $T \sim R_v$, we can see that whether or not $v$ can be $c_{new}$-active under a seed set $S$ in $T$ depends only on $S\cap V(G_v)$ because there is no path with all live edges in $T$ from $V \setminus V(G_v)$ to $v$ due to the construction in Alg. \ref{alg: rrtuple_v}, implying that $g(S,T)=g(S,R_v)$, and therefore $f_v(S)=\sum_{R_v \in \chi(v)} \sum_{T \sim R_v} \Pr[T] \cdot g(S,R_v)$. Finally, we have $\sum_{T \sim R_v} \Pr[T]=\Pr[R_v]$ for each $R_v \in \chi(v)$ of which the proof is elementary. Thus, we have \[\sum_{R_v \in \chi(v)} \sum_{T \sim R_v} \Pr[T] \cdot g(S,R_v)=\sum_{R_v \in \chi(v)} \Pr[R_v] \cdot g(S,R_v),\] which completes the proof.
\end{proof}

\subsection{Proof of Corollary \ref{coro: key}}
\label{proof: coro_key}
Note that we have $\E[g(S, \R)]=\sum_{v \in V}\frac{1}{n}\cdot \E[g(S ,\R_v)]$ for $g(S ,\R_v)$ as well as $\underline{g}(S ,\R_v)$ and $\overline{g}(S ,\R_v)$. Combining Theorem \ref{theorem: key_1}, we have 
\[\E[n\cdot \underline{g}(S, \R)] \leq \E[n\cdot g(S, \R)]=\sum_{v \in V} f_v(S) \leq E[n\cdot \overline{g}(S, \R)].\]
Finally, we have $\sum_{v \in V}{f_v(S)} ={f(S)}$.

\subsection{Proof of Lemma \ref{lemma: time}}
\label{proof: lemma_time}
Let $\E[\TIME_{\R_v}]$ be the expected running time of generating one RR-tuple of $v$ for a node $v \in V$, and we therefore have $\E[\TIME_{\R}]=\frac{1}{n}\sum_{v \in V}\E[\TIME_{\R_v}]$. According to Alg. \ref{alg: rrtuple_v}, it is clear that the running time of generating one RR-tuple of $v$ is bounded by the number the edges tested in Line 14 in Alg. \ref{alg: rrtuple_v}. For an edge $(u_1,u_2)$, it will be tested in Alg. \ref{alg: rrtuple_v} if and only if $u_2$ is in $\underline{S}_v$, and thus we have $\E[\TIME_{\R_v}]=\sum_{ (u_1,u_2)\in E}\E[\underline{g}(\{u_2\},\R_v)]$. As a result, 
\begin{align*}
\E[\TIME_{\R}]=\frac{1}{n}\sum_{v \in V}\sum_{ (u_1,u_2)\in E}\E[\underline{g}(\{u_2\},\R_v)]\\
=\sum_{ (u_1,u_2)\in E}\frac{1}{n}\sum_{v \in V}\E[\underline{g}(\{u_2\},\R_v)]
\leq m\cdot \max_{|S|=1}\E[\underline{g}(S ,\R)].
\end{align*}

\begin{figure}[tp]
	\begin{center}
		\includegraphics[width=0.40\textwidth]{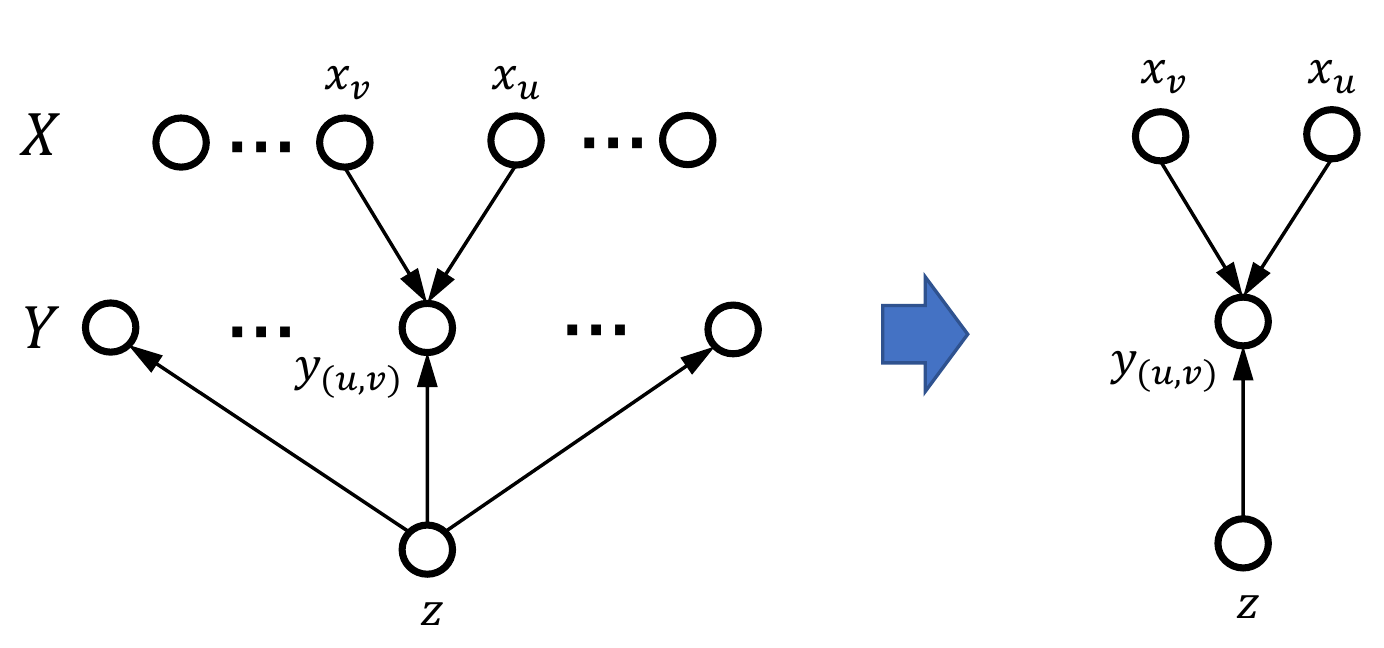} 
	\end{center} 
	\caption{Reduction from DkS to Problem \ref{problem: max_g}. An example of the RR-tuple associated with $y_{(u,v)}$.}
	\label{fig: reduction_max_g}
\end{figure}

\subsection{Proof of Lemma \ref{lemma: max_g}}
\label{proof: lemma_max_g}
Given an instance of the DkS problem, let us consider its corresponding graph Fig. \ref{fig: reduction_max} and construct an instance of Problem \ref{problem: max_g} as follows. For each node $y_{(u,v)}$, we construct an RR-tuple $R_{y_{(u,v)}}$ with subgraph $G_{y_{(u,v)}}$ consisting of the nodes and edges that are connected to $y_{(u,v)}$. An example is given in Fig. \ref{fig: reduction_max_g}. The settings of budget, $V_c$, and activation functions remain the same as that in Fig. \ref{fig: reduction_max}. For each $S \subseteq V_c$ with $|S|=k$, $g(S,R_{y_{(u,v)}})=1$ (i.e., $y_{(u,v)}$ is $c_{new}$-active under $S$ in $G_{y_{(u,v)}}$) if and only if both $x_u$ and $x_v$ are in $S$. Therefore, both instances have the same objective function value, and the rest of the analysis is similar to that of Lemma \ref{lemma: max_mim}.

\subsection{Proof of Lemma \ref{lemma: l}}
\label{proof: lemma_l}
It can be proved through the following two claims together with the union bound.
\begin{claim}
Eq. \ref{eq: c1} holds with probability at least $1-1/N$ if $l \geq l_1/f(S_{opt})$.
\end{claim}
\begin{proof}

For each $S \subseteq V_c$ with $|S|=k$, by Eq. (\ref{eq: G})~and Corollary \ref{coro: key}, $\Pr[\G(S,\Psi_l)-f(S)\geq \epsilon_1 \cdot \gamma(\Psi_l)\cdot f(S_{opt})]$ is rephrased as  \[\Pr\Big[{\sum_{\R \in \Psi_l}g(S, \R)}-l\cdot\E[ g(S ,\R)]\geq l\cdot \E[ g(S ,\R)]\cdot \delta\Big] \] where $\delta \define \frac{\epsilon_1 \cdot \gamma(\Psi_l)\cdot f(S_{opt})}{n \cdot \E[g(S ,\R)] }$. According to the Chernoff Bound, the above probability is no larger than $\exp(-\frac{l \cdot \E[g(S ,\R)] \cdot \delta^2}{2+\delta})$, i.e., 
\[\exp(-\frac{l \cdot \epsilon_1^2 \cdot \gamma^2(\Psi_l) \cdot f^2(S_{opt})}{2n^2\cdot\E[g(S ,\R)]+n+\epsilon_1\cdot \gamma(\Psi_l)\cdot f(S_{opt})}).\] Because $\E[g(S ,\R)]=\frac{f(S)}{n}\leq \frac{f(S_{opt})}{n}$ and $l \geq l_1/f(S_{opt})$, this probability is no larger than $\frac{1}{N\cdot \binom{n}{k}}$.
Because there are at most $\binom{n}{k}$ subsets with size $k$, by the union bound, with probability at least $1-1/N$ we have $\G(S,\Psi_l)-f(S)\leq \epsilon_1 \cdot \gamma(\Psi_l)\cdot f(S_{opt})$ for each $S \subseteq V_c$ with $|S|=k$.

\end{proof}

\begin{claim}
Eq. (\ref{eq: c2}) holds with probability at least $1-1/N$ if $l \geq l_2/f(S_{opt})$.
\end{claim}
\begin{proof}
According to Eq. (\ref{eq: G})~and Corollary \ref{coro: key}, $\Pr[\G(S_{opt},\Psi_l)-f(S_{opt})\leq -\epsilon_2 \cdot f(S_{opt})]$ is equal to \[\Pr[{\sum_{\R \in \Psi_l}g(S_{opt}, \R)}-l\cdot \E[ g(S_{opt} ,\R)]\leq -\epsilon_2 \cdot l\cdot \E[ g(S_{opt} ,\R)]]\] which is no larger than $\exp(-\frac{l \cdot \E[g(S_{opt} ,\R)] \cdot \epsilon_2^2}{2})$ according to the Chernoff bound. Since $l \geq l_2/f(S_{opt})$, it is further no larger than $1/N$.
\end{proof}

\subsection{Proof of Theorem \ref{theorem: main}}
\label{proof: theorem_main}
Due to Lemmas \ref{lemma: l}, \ref{lemma: lower_opt} and \ref{lemma: estimate}, the Eqs. (\ref{eq: c1}) and (\ref{eq: c2}) hold simultaneously with probability at least $1-3/N$. Eq. (\ref{eq: c3}) is guaranteed by Lemma \ref{lemma: sandwich}. With Eqs. (\ref{eq: c1})-(\ref{eq: c3}), the approximation ratio follows immediately from Lemma \ref{lemma: ratio}. Amplifying $N$ by a constant factor, we have the success probability of $1-1/N$. For those who are interested in a success probability of $1-n^{-l}$, the running time is $O(\frac{(m+n)(l+k)\ln n}{\epsilon^2})$.

\subsection{Proof of Lemma \ref{lemma: exp_upper}}
\label{proof: lemma: exp_upper}
\begin{proof}
According to Theorem \ref{theorem: key_1} and Corollary \ref{coro: key}, it suffices to prove that $\overline{g}(S, R_v)=g(S, R_v)$ for each RR-tuple $R_v=(G_v, \underline{S}_v, \overline{S}_v)$. If $\overline{g}(S, R_v)=1$, then $S \cap \overline{S}_v \neq \emptyset$, and $v$ will be $c_{new}$-active in $\M_{R_v}$ under $S$, which can be proved by inductively showing that all the nodes on the shortest path from $S \cap \overline{S}_v$ to $v$ will be $c_{new}$-active. So we have $g(S, R_v)=1$.  The other part that $g(S, R_v)=1$ implies $\overline{g}(S, R_v)=1$ is trivial.
\end{proof}





\bibliographystyle{IEEEtran}
\bibliography{bib_amo} 

\begin{thebibliography}{10}
\providecommand{\url}[1]{#1}
\csname url@samestyle\endcsname
\providecommand{\newblock}{\relax}
\providecommand{\bibinfo}[2]{#2}
\providecommand{\BIBentrySTDinterwordspacing}{\spaceskip=0pt\relax}
\providecommand{\BIBentryALTinterwordstretchfactor}{4}
\providecommand{\BIBentryALTinterwordspacing}{\spaceskip=\fontdimen2\font plus
\BIBentryALTinterwordstretchfactor\fontdimen3\font minus
  \fontdimen4\font\relax}
\providecommand{\BIBforeignlanguage}[2]{{%
\expandafter\ifx\csname l@#1\endcsname\relax
\typeout{** WARNING: IEEEtran.bst: No hyphenation pattern has been}%
\typeout{** loaded for the language `#1'. Using the pattern for}%
\typeout{** the default language instead.}%
\else
\language=\csname l@#1\endcsname
\fi
#2}}
\providecommand{\BIBdecl}{\relax}
\BIBdecl

\bibitem{li2018influence}
Y.~Li, J.~Fan, Y.~Wang, and K.-L. Tan, ``Influence maximization on social
  graphs: A survey,'' \emph{IEEE Transactions on Knowledge and Data
  Engineering}, vol.~30, no.~10, pp. 1852--1872, 2018.

\bibitem{zhang2014recent}
H.~Zhang, S.~Mishra, M.~T. Thai, J.~Wu, and Y.~Wang, ``Recent advances in
  information diffusion and influence maximization in complex social
  networks,'' \emph{Opportunistic Mobile Social Networks}, vol.~37, no. 1.1,
  p.~37, 2014.

\bibitem{sun2011survey}
J.~Sun and J.~Tang, ``A survey of models and algorithms for social influence
  analysis,'' in \emph{Social network data analytics}.\hskip 1em plus 0.5em
  minus 0.4em\relax Springer, 2011, pp. 177--214.

\bibitem{aslay2018influence}
C.~Aslay, L.~V. Lakshmanan, W.~Lu, and X.~Xiao, ``Influence maximization in
  online social networks,'' in \emph{Proceedings of the Eleventh ACM
  International Conference on Web Search and Data Mining}.\hskip 1em plus 0.5em
  minus 0.4em\relax ACM, 2018, pp. 775--776.

\bibitem{kempe2003maximizing}
D.~Kempe, J.~Kleinberg, and {\'E}.~Tardos, ``Maximizing the spread of influence
  through a social network,'' in \emph{Proceedings of the ninth ACM SIGKDD
  international conference on Knowledge discovery and data mining}.\hskip 1em
  plus 0.5em minus 0.4em\relax ACM, 2003, pp. 137--146.

\bibitem{datta2010viral}
S.~Datta, A.~Majumder, and N.~Shrivastava, ``Viral marketing for multiple
  products,'' in \emph{2010 IEEE International Conference on Data
  Mining}.\hskip 1em plus 0.5em minus 0.4em\relax IEEE, 2010, pp. 118--127.

\bibitem{li2014polarity}
D.~Li, Z.-M. Xu, N.~Chakraborty, A.~Gupta, K.~Sycara, and S.~Li, ``Polarity
  related influence maximization in signed social networks,'' \emph{PloS one},
  vol.~9, no.~7, p. e102199, 2014.

\bibitem{chen2011influence}
W.~Chen, A.~Collins, R.~Cummings, T.~Ke, Z.~Liu, D.~Rincon, X.~Sun, Y.~Wang,
  W.~Wei, and Y.~Yuan, ``Influence maximization in social networks when
  negative opinions may emerge and propagate,'' in \emph{Proceedings of the
  2011 siam international conference on data mining}.\hskip 1em plus 0.5em
  minus 0.4em\relax SIAM, 2011, pp. 379--390.

\bibitem{tong2018distributed}
G.~Tong, W.~Wu, and D.-Z. Du, ``Distributed rumor blocking with multiple
  positive cascades,'' \emph{IEEE Transactions on Computational Social
  Systems}, vol.~5, no.~2, pp. 468--480, 2018.

\bibitem{budak2011limiting}
C.~Budak, D.~Agrawal, and A.~El~Abbadi, ``Limiting the spread of misinformation
  in social networks,'' in \emph{Proceedings of the 20th international
  conference on World wide web}.\hskip 1em plus 0.5em minus 0.4em\relax ACM,
  2011, pp. 665--674.

\bibitem{tong2018misinformation}
G.~Tong, D.-Z. Du, and W.~Wu, ``On misinformation containment in online social
  networks,'' in \emph{Advances in Neural Information Processing Systems},
  2018, pp. 339--349.

\bibitem{chen2010scalable}
W.~Chen, C.~Wang, and Y.~Wang, ``Scalable influence maximization for prevalent
  viral marketing in large-scale social networks,'' in \emph{Proceedings of the
  16th ACM SIGKDD international conference on Knowledge discovery and data
  mining}.\hskip 1em plus 0.5em minus 0.4em\relax ACM, 2010, pp. 1029--1038.

\bibitem{borgs2014maximizing}
C.~Borgs, M.~Brautbar, J.~Chayes, and B.~Lucier, ``Maximizing social influence
  in nearly optimal time,'' in \emph{Proceedings of the twenty-fifth annual
  ACM-SIAM symposium on Discrete algorithms}.\hskip 1em plus 0.5em minus
  0.4em\relax SIAM, 2014, pp. 946--957.

\bibitem{tang2014influence}
Y.~Tang, X.~Xiao, and Y.~Shi, ``Influence maximization: Near-optimal time
  complexity meets practical efficiency,'' in \emph{Proceedings of the 2014 ACM
  SIGMOD international conference on Management of data}.\hskip 1em plus 0.5em
  minus 0.4em\relax ACM, 2014, pp. 75--86.

\bibitem{tang2015influence}
Y.~Tang, Y.~Shi, and X.~Xiao, ``Influence maximization in near-linear time: A
  martingale approach,'' in \emph{Proceedings of the 2015 ACM SIGMOD
  International Conference on Management of Data}.\hskip 1em plus 0.5em minus
  0.4em\relax ACM, 2015, pp. 1539--1554.

\bibitem{nguyen2016stop}
H.~T. Nguyen, M.~T. Thai, and T.~N. Dinh, ``Stop-and-stare: Optimal sampling
  algorithms for viral marketing in billion-scale networks,'' in
  \emph{Proceedings of the 2016 International Conference on Management of
  Data}.\hskip 1em plus 0.5em minus 0.4em\relax ACM, 2016, pp. 695--710.

\bibitem{lu2015competition}
W.~Lu, W.~Chen, and L.~V. Lakshmanan, ``From competition to complementarity:
  comparative influence diffusion and maximization,'' \emph{Proceedings of the
  VLDB Endowment}, vol.~9, no.~2, pp. 60--71, 2015.

\bibitem{feige1998threshold}
U.~Feige, ``A threshold of ln n for approximating set cover,'' \emph{Journal of
  the ACM (JACM)}, vol.~45, no.~4, pp. 634--652, 1998.

\bibitem{he2012influence}
X.~He, G.~Song, W.~Chen, and Q.~Jiang, ``Influence blocking maximization in
  social networks under the competitive linear threshold model,'' in
  \emph{Proceedings of the 2012 siam international conference on data
  mining}.\hskip 1em plus 0.5em minus 0.4em\relax SIAM, 2012, pp. 463--474.

\bibitem{fan2013least}
L.~Fan, Z.~Lu, W.~Wu, B.~Thuraisingham, H.~Ma, and Y.~Bi, ``Least cost rumor
  blocking in social networks,'' in \emph{2013 IEEE 33rd International
  Conference on Distributed Computing Systems}.\hskip 1em plus 0.5em minus
  0.4em\relax IEEE, 2013, pp. 540--549.

\bibitem{lin2015analyzing}
Y.~Lin and J.~C. Lui, ``Analyzing competitive influence maximization problems
  with partial information: An approximation algorithmic framework,''
  \emph{Performance Evaluation}, vol.~91, pp. 187--204, 2015.

\bibitem{tong2017efficient}
G.~Tong, W.~Wu, L.~Guo, D.~Li, C.~Liu, B.~Liu, and D.-Z. Du, ``An efficient
  randomized algorithm for rumor blocking in online social networks,''
  \emph{IEEE Transactions on Network Science and Engineering}, 2017.

\bibitem{li2018dominated}
H.~Li, L.~Pan, and P.~Wu, ``Dominated competitive influence maximization with
  time-critical and time-delayed diffusion in social networks,'' \emph{Journal
  of computational science}, vol.~28, pp. 318--327, 2018.

\bibitem{molinero2015cooperation}
X.~Molinero, F.~Riquelme, and M.~Serna, ``Cooperation through social
  influence,'' \emph{European Journal of Operational Research}, vol. 242,
  no.~3, pp. 960--974, 2015.

\bibitem{clark2011maximizing}
A.~Clark and R.~Poovendran, ``Maximizing influence in competitive environments:
  A game-theoretic approach,'' in \emph{International Conference on Decision
  and Game Theory for Security}.\hskip 1em plus 0.5em minus 0.4em\relax
  Springer, 2011, pp. 151--162.

\bibitem{fazeli2012game}
A.~Fazeli and A.~Jadbabaie, ``Game theoretic analysis of a strategic model of
  competitive contagion and product adoption in social networks,'' in
  \emph{2012 IEEE 51st IEEE Conference on Decision and Control (CDC)}.\hskip
  1em plus 0.5em minus 0.4em\relax IEEE, 2012, pp. 74--79.

\bibitem{tzoumas2012game}
V.~Tzoumas, C.~Amanatidis, and E.~Markakis, ``A game-theoretic analysis of a
  competitive diffusion process over social networks,'' in \emph{International
  Workshop on Internet and Network Economics}.\hskip 1em plus 0.5em minus
  0.4em\relax Springer, 2012, pp. 1--14.

\bibitem{apt2011diffusion}
K.~R. Apt and E.~Markakis, ``Diffusion in social networks with competing
  products,'' in \emph{International Symposium on Algorithmic Game
  Theory}.\hskip 1em plus 0.5em minus 0.4em\relax Springer, 2011, pp. 212--223.

\bibitem{li2015getreal}
H.~Li, S.~S. Bhowmick, J.~Cui, Y.~Gao, and J.~Ma, ``Getreal: Towards realistic
  selection of influence maximization strategies in competitive networks,'' in
  \emph{Proceedings of the 2015 ACM SIGMOD international conference on
  management of data}.\hskip 1em plus 0.5em minus 0.4em\relax ACM, 2015, pp.
  1525--1537.

\bibitem{shirazipourazad2012influence}
S.~Shirazipourazad, B.~Bogard, H.~Vachhani, A.~Sen, and P.~Horn, ``Influence
  propagation in adversarial setting: how to defeat competition with least
  amount of investment,'' in \emph{Proceedings of the 21st ACM international
  conference on Information and knowledge management}.\hskip 1em plus 0.5em
  minus 0.4em\relax ACM, 2012, pp. 585--594.

\bibitem{zhu2016minimum}
Y.~Zhu, D.~Li, and Z.~Zhang, ``Minimum cost seed set for competitive social
  influence,'' in \emph{IEEE INFOCOM 2016-The 35th Annual IEEE International
  Conference on Computer Communications}.\hskip 1em plus 0.5em minus
  0.4em\relax IEEE, 2016, pp. 1--9.

\bibitem{bozorgi2017community}
A.~Bozorgi, S.~Samet, J.~Kwisthout, and T.~Wareham, ``Community-based influence
  maximization in social networks under a competitive linear threshold model,''
  \emph{Knowledge-Based Systems}, vol. 134, pp. 149--158, 2017.

\bibitem{lu2013bang}
W.~Lu, F.~Bonchi, A.~Goyal, and L.~V. Lakshmanan, ``The bang for the buck: fair
  competitive viral marketing from the host perspective,'' in \emph{Proceedings
  of the 19th ACM SIGKDD international conference on Knowledge discovery and
  data mining}.\hskip 1em plus 0.5em minus 0.4em\relax ACM, 2013, pp. 928--936.

\bibitem{apt2014social}
K.~R. Apt and E.~Markakis, ``Social networks with competing products,''
  \emph{Fundamenta Informaticae}, vol. 129, no.~3, pp. 225--250, 2014.

\bibitem{pathak2010generalized}
N.~Pathak, A.~Banerjee, and J.~Srivastava, ``A generalized linear threshold
  model for multiple cascades,'' in \emph{2010 IEEE International Conference on
  Data Mining}.\hskip 1em plus 0.5em minus 0.4em\relax IEEE, 2010, pp.
  965--970.

\bibitem{borodin2010threshold}
A.~Borodin, Y.~Filmus, and J.~Oren, ``Threshold models for competitive
  influence in social networks,'' in \emph{International workshop on internet
  and network economics}.\hskip 1em plus 0.5em minus 0.4em\relax Springer,
  2010, pp. 539--550.

\bibitem{lin2015learning}
S.-C. Lin, S.-D. Lin, and M.-S. Chen, ``A learning-based framework to handle
  multi-round multi-party influence maximization on social networks,'' in
  \emph{Proceedings of the 21th ACM SIGKDD International Conference on
  Knowledge Discovery and Data Mining}.\hskip 1em plus 0.5em minus 0.4em\relax
  ACM, 2015, pp. 695--704.

\bibitem{litou2017influence}
I.~Litou, V.~Kalogeraki, and D.~Gunopulos, ``Influence maximization in a many
  cascades world,'' in \emph{2017 IEEE 37th International Conference on
  Distributed Computing Systems (ICDCS)}.\hskip 1em plus 0.5em minus
  0.4em\relax IEEE, 2017, pp. 911--921.

\bibitem{zarezade2017correlated}
A.~Zarezade, A.~Khodadadi, M.~Farajtabar, H.~R. Rabiee, and H.~Zha,
  ``Correlated cascades: Compete or cooperate,'' in \emph{Thirty-First AAAI
  Conference on Artificial Intelligence}, 2017.

\bibitem{garimella2017balancing}
K.~Garimella, A.~Gionis, N.~Parotsidis, and N.~Tatti, ``Balancing information
  exposure in social networks,'' in \emph{Advances in Neural Information
  Processing Systems}, 2017, pp. 4663--4671.

\bibitem{wang2017activity}
Z.~Wang, Y.~Yang, J.~Pei, L.~Chu, and E.~Chen, ``Activity maximization by
  effective information diffusion in social networks,'' \emph{IEEE Transactions
  on Knowledge and Data Engineering}, vol.~29, no.~11, pp. 2374--2387, 2017.

\bibitem{feige2002relations}
U.~Feige, ``Relations between average case complexity and approximation
  complexity,'' in \emph{Proceedings of the thiry-fourth annual ACM symposium
  on Theory of computing}.\hskip 1em plus 0.5em minus 0.4em\relax ACM, 2002,
  pp. 534--543.

\bibitem{alon2011inapproximability}
N.~Alon, S.~Arora, R.~Manokaran, D.~Moshkovitz, and O.~Weinstein,
  ``Inapproximability of densest k-subgraph from average case hardness, 2011,''
  \emph{Manuscript}, vol.~5, p.~19, 2011.

\bibitem{bhaskara2012polynomial}
A.~Bhaskara, M.~Charikar, V.~Guruswami, A.~Vijayaraghavan, and Y.~Zhou,
  ``Polynomial integrality gaps for strong sdp relaxations of densest
  k-subgraph,'' in \emph{Proceedings of the twenty-third annual ACM-SIAM
  symposium on Discrete Algorithms}.\hskip 1em plus 0.5em minus 0.4em\relax
  SIAM, 2012, pp. 388--405.

\bibitem{manurangsi2017almost}
P.~Manurangsi, ``Almost-polynomial ratio eth-hardness of approximating densest
  k-subgraph,'' in \emph{Proceedings of the 49th Annual ACM SIGACT Symposium on
  Theory of Computing}.\hskip 1em plus 0.5em minus 0.4em\relax ACM, 2017, pp.
  954--961.

\bibitem{miettinen2008positive}
P.~Miettinen, ``On the positive--negative partial set cover problem,''
  \emph{Information Processing Letters}, vol. 108, no.~4, pp. 219--221, 2008.

\bibitem{leskovec2007cost}
J.~Leskovec, A.~Krause, C.~Guestrin, C.~Faloutsos, C.~Faloutsos, J.~VanBriesen,
  and N.~Glance, ``Cost-effective outbreak detection in networks,'' in
  \emph{Proceedings of the 13th ACM SIGKDD international conference on
  Knowledge discovery and data mining}.\hskip 1em plus 0.5em minus 0.4em\relax
  ACM, 2007, pp. 420--429.

\bibitem{report}
A.~Authors, ``On multi-cascade influence maximization: Model, hardness and
  algorithmic framework,'' Tech. Rep., 2019.

\bibitem{vazirani2013approximation}
V.~V. Vazirani, \emph{Approximation algorithms}.\hskip 1em plus 0.5em minus
  0.4em\relax Springer Science \& Business Media, 2013.

\bibitem{liu2013combining}
H.~Liu, J.~He, T.~Wang, W.~Song, and X.~Du, ``Combining user preferences and
  user opinions for accurate recommendation,'' \emph{Electronic Commerce
  Research and Applications}, vol.~12, no.~1, pp. 14--23, 2013.

\bibitem{zeng2009social}
F.~Zeng, L.~Huang, and W.~Dou, ``Social factors in user perceptions and
  responses to advertising in online social networking communities,''
  \emph{Journal of interactive advertising}, vol.~10, no.~1, pp. 1--13, 2009.

\bibitem{bodendorf2009detecting}
F.~Bodendorf and C.~Kaiser, ``Detecting opinion leaders and trends in online
  social networks,'' in \emph{Proceedings of the 2nd ACM workshop on Social web
  search and mining}.\hskip 1em plus 0.5em minus 0.4em\relax ACM, 2009, pp.
  65--68.

\bibitem{leskovec2015snap}
J.~Leskovec and A.~Krevl, ``$\{$SNAP Datasets$\}$:$\{$Stanford$\}$ large
  network dataset collection,'' 2015.

\bibitem{de2013anatomy}
M.~De~Domenico, A.~Lima, P.~Mougel, and M.~Musolesi, ``The anatomy of a
  scientific rumor,'' \emph{Scientific reports}, vol.~3, p. 2980, 2013.

\bibitem{asahiro2002complexity}
Y.~Asahiro, R.~Hassin, and K.~Iwama, ``Complexity of finding dense subgraphs,''
  \emph{Discrete Applied Mathematics}, vol. 121, no. 1-3, pp. 15--26, 2002.

\end{thebibliography}

\end{document}